\documentclass[journal,letter,twoside]{IEEEtran}
\usepackage{cite}
\usepackage{amsmath}
\usepackage{amssymb}
\usepackage{amsfonts}
\usepackage{bm}
\usepackage{ftnxtra}
\usepackage{fnpos}
\usepackage{textcomp}
\usepackage{comment}
\usepackage{verbatim}
\usepackage{wasysym}
\usepackage{setspace}
\usepackage{booktabs}
\def\smallerspacecaption{\vspace{-2mm}}

\usepackage{floatrow}
\usepackage[caption=false,font=footnotesize]{subfig}
\usepackage[table, dvipsnames]{xcolor}
\usepackage{fancyhdr}
\usepackage[implicit=true,hidelinks]{hyperref}
\usepackage{graphicx}
\usepackage{url}
\usepackage{verbatim}
\usepackage{multirow}
\usepackage{hhline}
\usepackage[us]{datetime}
\usepackage{paralist}
\usepackage{color}
\usepackage{soul}
\usepackage{amsthm}

\newtheorem{theorem}{Theorem}
\newtheorem{definition}{Definition}
\usepackage{mathtools}

\usepackage{stfloats}
\graphicspath{{figures/}}

\usepackage{pifont}

\definecolor{Gray}{gray}{0.6}
\definecolor{LightGray}{gray}{0.9}
\definecolor{LighterGray}{gray}{0.7}
\newcolumntype{a}{>{\columncolor{Gray}}c}
\newcolumntype{b}{>{\columncolor{LightGray}}c}
\newcolumntype{d}{>{\columncolor{LighterGray}}c}
\definecolor{cadmiumgreen}{rgb}{0.0, 0.42, 0.24}

\newdimen\arrayruleHwidth
\makeatletter
\def\Hline{\noalign{\ifnum0=`}\fi\hrule \@height \arrayruleHwidth
\futurelet \@tempa\@xhline}
\makeatother
\newcolumntype{P}[1]{>{\centering\arraybackslash}p{#1}}
\makeatletter
\def\blfootnote{\xdef\@thefnmark{}\@footnotetext}
\makeatother
\shortdate
\settimeformat{ampmtime}
\ifCLASSINFOpdf
\else
\fi
\usepackage[ruled, vlined, norelsize, linesnumbered]{algorithm2e}
\SetKwInput{KwInput}{Input}
\SetKwInput{KwOutput}{Output}
\SetKwInput{KwInit}{Initialization}
\usepackage{adjustbox}
\SetAlFnt{\sf}
\usepackage[table]{xcolor}
\usepackage{multirow}
\usepackage{alltt}
\usepackage{tcolorbox}
\usepackage{diagbox}

\newcommand{\llbox}[1]{
    \begin{tcolorbox}[
    sharp corners,
    boxrule=1pt,
    boxsep = 0pt,
    left = 5pt,
    right = 5pt,
    colback=white!10!white,
    colframe=black!50!black,
    ]
    {#1}
    \end{tcolorbox}
}

\usepackage{thm-restate}

\usepackage{enumitem}
\usepackage{tikz}

\begin{document}

\newcommand{\maintitle}{Hide \& Seek}
\newcommand{\subtitle}{Seeking the (Un)-Hidden key in Provably-Secure Logic Locking Techniques}

\newcommand{\thetitle}{\maintitle: \subtitle}

\title{\thetitle}

\author{Satwik~Patnaik,~\IEEEmembership{Member,~IEEE}, Nimisha~Limaye,~\IEEEmembership{Graduate~Student~Member,~IEEE,} and Ozgur~Sinanoglu,~\IEEEmembership{Senior~Member,~IEEE}

\thanks{Manuscript received March 17, 2022; revised August 9, 2022; accepted September 3, 2022.
The associate editor coordinating the review of this
manuscript and approving it for publication was Prof.\ Ulrich R\"uhrmair.
(\textit{Corresponding authors: Satwik~Patnaik and Nimisha~Limaye.})
}
\IEEEcompsocitemizethanks{\IEEEcompsocthanksitem Satwik~Patnaik is with the Department of Electrical and Computer Engineering, Texas A\&M University, College Station, TX 77843, USA (e-mail: satwik.patnaik@tamu.edu).
\IEEEcompsocthanksitem Nimisha~Limaye is with the Department of Electrical and Computer Engineering, Tandon School of Engineering, New York University, Brooklyn, NY 11201, USA (email: nimisha.limaye@nyu.edu).\protect
\IEEEcompsocthanksitem Ozgur~Sinanoglu is with the Division of Engineering, New York University Abu Dhabi, Abu Dhabi 129188, UAE (email: ozgursin@nyu.edu).\protect
}
\thanks{Digital Object Identifier 10.1109/TIFS.2022.XXXXXXX}
}

\markboth{IEEE Transactions on Information Forensics and Security}
{Patnaik \MakeLowercase{\textit{et al.\ }}: Hide \& Seek: Seeking the (Un)-Hidden key in Provably-Secure Logic Locking Techniques}

\IEEEtitleabstractindextext{
\begin{abstract}
Logic locking is a holistic countermeasure that protects an integrated circuit (IC) from hardware-focused threats such as piracy of design intellectual property and unauthorized overproduction throughout the globalized IC supply chain.
Out of the several techniques proposed by the hardware security community, provably-secure logic locking (PSLL) has acquired a foothold due to its algorithmic and provable-security guarantees.
However, the security of these techniques are regularly questioned by attackers that exploit the vulnerabilities arising from the underlying hardware implementation.
Unfortunately, such attacks (i)~are predominantly specific to locking techniques and (ii)~lack generality and scalability. 
This leads to a plethora of attacks and researchers, especially defenders, find it challenging to ascertain the security of newly developed PSLL techniques.
Additionally, there is no public repository of locked circuits that attackers can use to benchmark (and compare) their developed attacks.

Driven by these challenges, we aim to develop a generalized attack that can recover the secret key across a breadth of PSLL techniques.
To that end, we first categorize the existing PSLL techniques into two generic categories.
Then, we extract functional and structural properties depending on the underlying hardware construction of the PSLL techniques and develop two attacks based on the concepts
of VLSI testing and Boolean transformations.
We evaluate our attacks on 30,000 locked circuits across 14 PSLL techniques, including nine unbroken techniques. 
Our attacks successfully recover the secret key (100\% accuracy) for all the considered techniques. 
Further, our experimentation across different (i)~technology libraries, (ii)~commercial and academic synthesis tools, and (iii)~logic optimization settings provide several interesting insights.
For instance, our attacks can recover the secret key by only using the locked circuit when an academic synthesis tool is used. 
Additionally, designers can use our attacks as a verification tool to ascertain the lower-bound security achieved by hardware implementations.
Finally, we shall release our artifacts, which could help foster the development of future attacks and defenses in PSLL domain.
\end{abstract}

}

\maketitle

\IEEEdisplaynontitleabstractindextext
\IEEEpeerreviewmaketitle

\renewcommand{\headrulewidth}{0.0pt}
\thispagestyle{fancy}
\lhead{}
\rhead{}
\chead{To Appear in IEEE Transactions on Information Forensics and Security (TIFS), 2022}
\cfoot{}

\section{Introduction}
\label{sec:introduction}

\IEEEPARstart{T}{he} continual miniaturization of integrated circuit (IC) technology nodes have exacerbated the costs of commissioning state-of-the-art foundries~\cite{tsmc3nm}. 
Designs are regularly outsourced to potentially untrustworthy foundries, and as a result, several hardware-focused threats have emerged, ranging from piracy of design intellectual property (IP) and unauthorized overproduction of ICs to insertion of malicious logic~\cite{rostami14}.

\textbf{Logic locking} is a holistic countermeasure that protects an IC from several hardware-focused 
threats such as reverse-engineering, piracy of design IP, and unauthorized
overproduction throughout the IC supply chain~\cite{yasin_CCS_2017}.
Logic locking transforms the original circuit by incorporating additional logic (key-gates) controlled by a secret key. 
As a result of inserting key-gates, a locked circuit includes additional inputs, referred to as key-inputs apart from regular primary inputs.
The secret key is stored in a
tamper-proof memory and securely programmed by a trusted facility (e.g., a design house) \textit{after} the fabrication and testing of the ICs.
The application of the correct key ensures the locked circuit functions correctly (for all input patterns), while an incorrect key renders the locked circuit to produce corrupted outputs.
\textit{The security guarantees offered by logic locking techniques are contingent on
the inability of an attacker to
recover the secret key.}
Prior combinational logic locking techniques
focused on (i)~finding suitable locations for key-gate insertion~\cite{epic,JV_DAC_2012}, and (ii)~exploring different key-gates (e.g., multiplexers~\cite{JV-Tcomp-2013}).

\textbf{Input/Output-based attacks:} The Boolean Satisfiability-based attack (commonly known as SAT-based attack in the logic locking community)~\cite{subramanyan15}
broke all known logic locking techniques in 2015.
The attack uses a SAT solver to generate \textit{distinguishing input patterns (DIPs)}---these input patterns enable the elimination of incorrect keys from the key search space.
The DIPs, along with output responses from a working chip (a.k.a. \textit{oracle}), iteratively eliminate incorrect keys, resulting in the recovery of the secret key.
Subsequently, researchers developed approximate-based attacks (\textit{AppSAT}~\cite{shamsi2017appsat} and \textit{Double DIP}~\cite{shen2017double}) that relax the exactness constraint in the SAT-based attack to yield an approximate key.
All the aforementioned attacks utilize input/output (I/O) pairs from an oracle and thus are called \textit{I/O-based attacks}.

\textbf{I/O-based attack resilient locking:} The logic locking community proposed several techniques to thwart I/O-based attacks.
These can be categorized under (i)~point function-based locking,\footnote{Also known as provably-secure logic locking, more details in \S\ref{sec:PSLL}.} (ii)~SAT-hard locking, (iii)~cyclic locking, and (iv)~scan locking.
Concerning (i), researchers proposed augmenting the original circuit with logic structures (e.g., point-functions) that ensure I/O-based attacks can prune out exactly one incorrect key in every attack iteration~\cite{SARLock_host_2016,xie16_SAT}. 
Adopting this construction necessitates I/O-based attacks to query an exponential number of input patterns (regarding key-size) to recover the secret key.
The techniques under (ii), \textit{i.e.,} SAT-hard locking embed
structures (e.g., look-up tables, multipliers)
that realize complicated SAT formulas, which increase the time taken per attack iteration~\cite{kamali2019full}. 
The techniques under (iii), \textit{i.e.,} cyclic locking, instantiate feedback cycles to thwart I/O-based attacks~\cite{shamsi2017cyclic}. 
Inserting cycles inhibits the locked circuit from being modeled as a directed acyclic graph, an important requirement for most I/O-based attacks. 
Finally, the techniques under (iv), \textit{i.e.,} scan locking, obfuscate the scan data, limiting the controllability and observability of internal nets~\cite{karmakar2018encrypt}, a prime enabler behind the success of I/O-based attacks.
We consider the techniques in (i)~because of their algorithmic security guarantees in thwarting I/O-based attacks.

\subsection{Arms Race Between Attackers and Defenders in PSLL}
\label{sec:attacks_defenses_PSLL}

SARLock~\cite{SARLock_host_2016} and Anti-SAT~\cite{xie16_SAT} were the first techniques to thwart I/O-based attacks.
These techniques add point-functions to the original circuit, thereby necessitating an attacker to apply exponential input patterns (regarding key-size) to recover the secret key.
However, both techniques were thwarted by \textit{bypass} attack~\cite{bypass_ches_2017} and removal attacks.\footnote{Removal attacks identify (and isolate) the protection logic and remove it from the locked circuit. 
Removing the protection logic yields the original circuit to an attacker. 
Although we acknowledge the existence of removal attacks, we restrict the discussion to key-recovery attacks in this work.}

Researchers adopted the paradigm of corrupt and correct-based PSLL techniques (also known as stripped-functionality logic locking (SFLL)) where designers enforce controlled corruption for user-specified input pattern(s) by hard-coding them using point-functions. 
These errors are corrected when the correct key is provided through a key-controlled unit~\cite{yasin_CCS_2017}.
However, attackers successfully recovered the secret key through structural and functional analysis~\cite{yang2019stripped,FALL}.
A logic removal-based locking approach (SFLL-rem)~\cite{sengupta2020truly} demonstrated resilience against attackers during a global logic locking competition.
However, this technique has been recently circumvented, where researchers demonstrated the intricacies between logic synthesis and logic locking~\cite{SPI_USENIX}.
Researchers proposed improvements over Anti-SAT (viz., CASLock~\cite{caslock_ches_2019}), which thwarted the bypass attack.
However, researchers have demonstrated attacks that recovered the secret key~\cite{CASUnlock}.

\subsection{Motivation and Research Challenges}
\label{sec:research_challenges}

As evidenced from the previous sub-section, there has been an arms race between attackers and defenders. 
Although a plethora of attacks have been proposed; unfortunately, most attacks target specific PSLL
techniques, as evidenced next.
For instance, the bypass attack~\cite{bypass_ches_2017} demonstrated vulnerabilities in SARLock~\cite{SARLock_host_2016} and Anti-SAT~\cite{xie16_SAT} but could not challenge the security of SFLL techniques~\cite{yasin_CCS_2017,sengupta2020truly}.
The FALL~\cite{FALL} and SFLL-hd-unlocked~\cite{yang2019stripped} attacks were successful in recovering the secret key from variants of 
SFLL-HD but did not apply to SFLL-flex~\cite{yasin_CCS_2017} and SFLL-rem~\cite{sengupta2020truly}. 
The attacks proposed in~\cite{CASUnlock} broke the security guarantees of CASLock~\cite{caslock_ches_2019} 
and Anti-SAT~\cite{xie16_SAT} but did not consider other PSLL techniques such as SFLL-flex~\cite{yasin_CCS_2017}, SFLL-rem~\cite{sengupta2020truly}, and corrupt-and-correct (CAC)~\cite{CAC} (to name a few).
The sparse prime implicant (SPI) attack~\cite{SPI_USENIX} recovered the secret key from SFLL-rem and SFLL-HD$^{0}$ but did not consider several unbroken PSLL techniques such as CAC~\cite{CAC}, diversified tree logic (DTL)~\cite{CAC}, Strong Anti-SAT (SAS)~\cite{liu2020strong}, and variants of Gen-Anti-SAT~\cite{zhou2021generalized}. 
\textit{Despite the existence of
all these attacks, nine PSLL techniques
have not been tackled from the standpoint of key-recovery attacks.\footnote{We refer interested readers to our work in~\cite{limaye2022valkyrie} where we showcased removal attacks.
However, as stated previously (footnote 2), this work aims to recover the secret key from the hardware implementation of PSLL techniques.}}
The aforementioned discussion highlights that state-of-the-art key-recovery attacks are (i)~locking technique specific (\textit{i.e.,} the generality is limited) and (ii)~unable to challenge the security guarantees of recent PSLL techniques.
This leads to our first research challenge.

\llbox{\textbf{RC1:} Can we formulate generalized attacks that recover the secret key from the hardware implementation of unbroken and broken PSLL techniques?}

The  
I/O-based attacks demonstrated that
logic locking techniques having a key-size of $k$ does not necessarily 
imply \textit{k}-bit security.
The actual security level depends on the mathematical primitive and the scheme construction~\cite{guo2018introduction}.
Although PSLL techniques are mathematically sound (assuming that DIPs are chosen uniformly at random and are non-repeated, I/O-based attacks require $2^k$ queries to an oracle) in recovering a $k$-bit key), the hardware implementation 
of these techniques leave structural vulnerabilities
that attackers exploit to recover the secret key. 
\textit{Hence, there is a requirement for a security framework that 
informs a designer regarding the lower-bound security-level attained by the hardware implementation of PSLL techniques.}
This leads to our second research challenge.

\llbox{
\textbf{RC2:} Can we develop a security framework that informs designers regarding the lower-bound security-level attained by the hardware implementation of PSLL techniques?
}

\subsection{Our Research Contributions}
\label{sec:scope_of_work}

\begin{table*}[ht]
\caption{Efficacy of our proposed key-recovery attacks against the state-of-the-art attacks}
\footnotesize
\label{tab:compare_KR}
\setlength{\tabcolsep}{0.43mm}
\begin{tabular}{cccccccccccccccc}
\hline
\multirow{2}{*}{\backslashbox{\textbf{Attack}}{\textbf{Defense}}} 
& 
\multirow{2}{*}{\textbf{SARLock}} 
& 
\multirow{2}{*}{\textbf{Anti-SAT}} 
&
\multirow{2}{*}{\textbf{SFLL-HD$^0$}} 
& 
\multirow{2}{*}{\textbf{SFLL-flex}} 
& 
\multirow{2}{*}{\textbf{SFLL-rem}} 
& 
\multirow{2}{*}{\textbf{CASLock}} 
& 
\multirow{2}{*}{\textbf{ECE}} 
& 
\multirow{2}{*}{\textbf{SAS}} 
& 
\multicolumn{2}{c}{\textbf{Gen-Anti-SAT}} 
& 
\multirow{2}{*}{\textbf{CAC}} 
& 
\multicolumn{3}{c}{\textbf{DTL}} 
\\ 
\cline{10-11} 
\cline{13-15} 
&  
&
&
&
&
&  
&  
&
&
\textbf{Comp.} 
& 
\textbf{Non-comp.} 
&  
& 
\textbf{SARLock} 
& 
\textbf{Anti-SAT} 
& 
\textbf{CAC} \\ 
\hline 
\hline

\textbf{SAT~\cite{subramanyan15}} & \textbf{{\Circle}} & \textbf{\Circle} &
\textbf{\Circle} & \textbf{\Circle} & \textbf{\Circle} & \textbf{\Circle} & \textbf{\Circle} &
\textbf{\Circle} & \textbf{\Circle} & \textbf{\Circle} & \textbf{\Circle} & \textbf{\Circle} &
\textbf{\Circle} &
\textbf{\Circle} 
\\ \hline

\textbf{Bypass~\cite{bypass_ches_2017}} & \textbf{\CIRCLE} & \textbf{\CIRCLE} &
\textbf{\Circle} & \textbf{\Circle} & \textbf{\Circle} & \textbf{\Circle} & \textbf{\Circle} & \textbf{\Circle} &
\textbf{\Circle} & \textbf{\Circle} & \textbf{\Circle} & \textbf{\Circle} &
\textbf{\Circle} &
\textbf{\Circle} 
\\ \hline

\textbf{SFLL-hd-unlocked~\cite{yang2019stripped}} & \textbf{\Circle} & \textbf{\Circle} &
\textbf{\CIRCLE} & \textbf{\Circle} & \textbf{\Circle} & \textbf{\Circle} & \textbf{\Circle} & 
\textbf{\Circle} &
\textbf{\Circle} & \textbf{\Circle} & \textbf{\Circle} & \textbf{\Circle} &
\textbf{\Circle} &
\textbf{\Circle} 
\\ \hline

\textbf{FALL~\cite{FALL}} & \textbf{\Circle} & \textbf{\Circle} &
\textbf{\CIRCLE} & \textbf{\Circle} & \textbf{\Circle} & \textbf{\Circle} & \textbf{\Circle} & \textbf{\Circle} & \textbf{\Circle} &
\textbf{\Circle} & \textbf{\Circle} & \textbf{\Circle} &
\textbf{\Circle} &
\textbf{\Circle} 
\\ \hline

\textbf{SPI~\cite{SPI_USENIX}} & \textbf{\CIRCLE} 
& \textbf{\CIRCLE} 
& 
\textbf{\CIRCLE} & \textbf{\Circle} & \textbf{\CIRCLE} & \textbf{\Circle} & \textbf{\Circle} & \textbf{\Circle} & \textbf{\Circle} & \textbf{\Circle} &
\textbf{\Circle} & \textbf{\Circle} &
\textbf{\Circle} &
\textbf{\Circle} 
\\ \hline

\textbf{CASUnlock~\cite{CASUnlock}} & \textbf{\Circle} & \textbf{\CIRCLE} &
\textbf{\Circle} & \textbf{\Circle} & \textbf{\Circle} & \textbf{\CIRCLE} &
\textbf{\Circle} &
\textbf{\Circle} & \textbf{\Circle} & \textbf{\Circle} & \textbf{\Circle} & \textbf{\Circle} &
\textbf{\Circle} &
\textbf{\Circle} 
\\ \hline

\textbf{This Work} & \textbf{\CIRCLE} & \textbf{\CIRCLE} &
\textbf{\CIRCLE} & \textbf{\CIRCLE} & \textbf{\CIRCLE} & \textbf{\CIRCLE} & \textbf{\CIRCLE} & \textbf{\CIRCLE} & \textbf{\CIRCLE} & \textbf{\CIRCLE} &
\textbf{\CIRCLE} & 
\textbf{\CIRCLE} &
\textbf{\CIRCLE} &
\textbf{\CIRCLE} 
\\ \hline
\end{tabular}
\\
[1mm]
\CIRCLE Successful attack \hspace{5mm}
\Circle Unsuccessful/undocumented attack
\end{table*}

Our work addresses the aforementioned research challenges by developing attacks that successfully recover the secret key from the hardware implementation of PSLL techniques.
Our attacks (i)~apply to a breadth of PSLL techniques, (ii)~successfully recover the secret key for five previously broken and nine unbroken PSLL techniques, (ii)~support industry-adopted Verilog format, (iv)~do not require \textit{a-priori} information, \textit{i.e.,} the functionality of the locked design, (v)~are agnostic to the choice of synthesis tool, synthesis commands, technology libraries, and choice of logic gates used to realize the hardware implementation, (vi)~are scalable to large-scale designs and key-sizes, and (vii)~can be utilized as a diagnostic tool by designers to ascertain the lower bound security-level attained
by the hardware implementation of PSLL techniques.
The primary contributions of our work are as follows.

\begin{itemize}[leftmargin=*]

\item We conceptualize and implement two 
generalized attacks that recover
the secret key from the hardware implementation of 14 PSLL techniques, \textbf{including nine unbroken techniques.}
Our attacks leverage structural and functional properties stemming from the underlying construction of PSLL techniques coupled with VLSI testing principles and Boolean transformations (\S\ref{sec:hardcoded_PSLL} and \S\ref{sec:non_hardcoded_PSLL}).
Our attacks apply to a breadth of PSLL techniques, as opposed to other attacks that have been PSLL technique-specific (Table~\ref{tab:compare_KR}).

\item We demonstrate the efficacy of our key-recovery attacks by performing experiments across 30,000 locked circuits. 
\textbf{Our attacks achieve 100\% accuracy in recovering the secret key for all locked circuits.}
\textit{Our attacks are agnostic to the choice of (i)~synthesis tool, (ii)~synthesis commands, (iii)~technology libraries, and (iv)~logic gates used during synthesis}.
In short, our analysis illustrates the inadequacies of academic and commercial CAD tools used for realizing hardware implementation of PSLL techniques (\S\ref{sec:results}).

\item We present interesting insights from our attacks (\S\ref{sec:findings_KR_attack}) and suggest that security-enforcing designers and developers of PSLL techniques utilize our attacks as
a diagnostic tool (\S\ref{sec:diagnostic_tool}).
Using our attacks, designers can ascertain the lower-bound security level (within a few minutes) achieved by the hardware implementation of a newly developed PSLL technique.
Our analysis reveals that structural security of hard-coded\footnote{Hard-coded PSLL techniques are explained in detail in \S\ref{sec:PSLL}.} PSLL techniques depend on the choice of the secret key, which calls for further investigation on the secure hardware implementation of PSLL techniques.

\item Finally, we shall release our artifacts to foster the development of new attacks and defense techniques.

\end{itemize}
\section{Background and Preliminaries}
\label{sec:background}

\subsection{Notations and Definitions}
\label{sec:notations}

\noindent\textbf{Notations.} Let \begin{math} \mathbb{B} \end{math} = \begin{math} \{0,1\} \end{math} be the Boolean domain.
The notation $\{x_0, x_1, x_2\}$ denotes a set of elements $x_0$, $x_1$, and $x_2$. 
We denote a set $A$ as a subset of set $B$ as $A \subseteq B$.
We use italics to denote variables such as primary inputs or $PI$ = \{$I_i$\}, where $i$ $\in$ $\{0 \dots n-1\}$, primary outputs or $PO$ = \{$O_i$\}, where $i$ $\in$ $\{0 \dots m-1\}$, protected input ports or $PIP \subseteq PI$, protected output ports or $POP \subseteq PO$, key-inputs or $KI$ = \{$K_i$\}, where $i$ $\in$ $\{0 \dots k-1\}$, wires (edges) or $E$ $\in$ \{$n_0$, $n_1$, $\dots$, $n_{p-1}$\}, and gates (vertices) or $V$ $\in$ \{$v_0$, $v_1$, $\dots$, $v_{q-1}$\}. 
Constant pattern is denoted as $\mathtt{PP}$ for protected pattern, $\mathtt{TP}$ for test pattern, $\mathtt{f}$ for fault value, and $\mathtt{K}$ for secret key value.
A pattern value can be denoted as $\langle p_0,p_1,p_2\rangle$, where \{$p_0,p_1,p_2$\} $\in$ \{0,1\}.
Notation $a \land b$ denotes conjunction (AND) of $a$ and $b$, $a \lor b$ denotes disjunction (OR), $a \oplus b$ denotes exclusive or (XOR), and $\lnot a$ denotes logical negation (NOT).

A combinational circuit 
\begin{math} \mathcal{C}_{orig} \end{math}
is a directed acyclic graph (DAG) having $n$ $PIs$ and $m$ $POs$ implementing a Boolean function $F:PI\rightarrow PO$, where $PI = \{0,1\}^n$ and $PO = \{0,1\}^m$.
It contains $p$ wires and $q$ gates.
A logic locking technique $\mathcal{L}$  locks \begin{math} \mathcal{C}_{orig} \end{math} with a secret key $\mathtt{K}$ to obtain a locked circuit \begin{math} \mathcal{C}_{lock} \end{math}.
\begin{math} \mathcal{C}_{lock} \end{math} is $L: PI\times KI \rightarrow PO$. $KI=\{0,1\}^{|K|}$, where $|K|$ denotes the cardinality of $KI$ and is called the key-size.
\begin{math} \mathcal{C}_{lock} \end{math} is fabricated by an \textit{untrustworthy} foundry and converted into a chip \begin{math} \mathbb{C}_{lock} \end{math}.
After \begin{math} \mathbb{C}_{lock} \end{math} is tested and packaged, a \textit{trustworthy} facility (e.g., design house) activates the chip by loading the tamper-proof memory with the correct key $\mathtt{K}$ to obtain an activated chip \begin{math} \mathbb{C}_{act} \end{math}.
This activated chip is also known as an oracle in the logic locking community.
$\mathcal{A}^{\mathbb{S}}$ denotes an attacker $\mathcal{A}$ following an attack strategy $\mathbb{S}$. 
The goal of an attacker $\mathcal{A}^{\mathbb{S}}$ is to recover a key $\mathtt{K}_{rec}$ such that \begin{math} \mathbb{C}_{lock} (i,\mathtt{K}_{rec}) \end{math} = \begin{math} \mathcal{C}_{orig} (i) \end{math}, \begin{math} \forall i \in I \end{math}.
Upon a successful key-recovery attack, the recovered circuit \begin{math} \mathcal{C}_{rec} \end{math} is functionally equivalent to the original circuit \begin{math} \mathcal{C}_{orig} \end{math}, \textit{i.e.,} \begin{math} \mathcal{C}_{rec} (i) \end{math} = \begin{math} \mathcal{C}_{orig} (i) \end{math},
\begin{math} \forall i \in I \end{math}.

\begin{definition}
\textbf{Algorithmic security~\cite{yasin_CCS_2017}.} A logic locking technique $\mathcal{L}$ is $\alpha$-secure against an attacker
$\mathcal{A}^{\mathbb{IO}}$ making a polynomial number of I/O queries $q(\alpha)$ to a working chip \begin{math} \mathbb{C}_{act} \end{math}, if he/she cannot reconstruct \begin{math} \mathcal{C}_{rec} \end{math} correctly with a probability $P_{succ}$ greater than $\frac{q(\alpha)}{2^{\alpha}}$. 
\end{definition}

\begin{definition}
\textbf{Structural security~\cite{SPI_USENIX}.} A logic locking technique $\mathcal{L}$ is $\beta$-secure against an attacker $\mathcal{A}^{\mathbb{S}}$ performing white-box structural analysis of the locked circuit \begin{math} \mathcal{C}_{lock} \end{math}, if the probability to recover the secret is no greater than $\frac{1}{\beta}$.
\end{definition}

\subsection{Threat Model}
\label{sec:threat_model}

Now, we discuss the capabilities of an attacker \begin{math}\mathcal{A}\end{math}, motivation for the attack, and different attack settings.
An attacker reverse-engineers the GDSII\footnote{An industry-standard binary file format used by designers for sharing layout-level information (pertaining to an IC) with foundries.} information and extracts the gate-level netlist of the locked circuit \begin{math} \mathcal{C}_{lock} \end{math}.
In addition, she has access to a test pattern generation (TPG) tool \begin{math} \mathcal{T} \end{math} and a synthesis tool \begin{math} \mathcal{S} \end{math}. 
She also has access to a working copy of the chip \begin{math} \mathbb{C}_{act} \end{math}
(a.k.a. \textit{oracle}) with the secret key loaded in the tamper-proof memory.
Note that under our threat model (which is consistent and agreed upon by researchers in the logic locking community), (i)~access to \begin{math} \mathcal{C}_{lock} \end{math} is unrestricted and (ii)~access to
\begin{math} \mathbb{C}_{act} \end{math} is restricted, \textit{i.e.,} an attacker can only use \begin{math} \mathbb{C}_{act} \end{math}
to make oracle queries (an attacker can apply input pattern(s) and observe the output response(s).
Additionally, we assume that an attacker cannot insert Trojans or probe the tamper-proof memory or the key-registers to recover the secret key.
Furthermore, an attacker (i)~knows the type of PSLL technique implemented by the defender, and (ii)~can distinguish between $PIs$ and $KIs$.
All the assumptions are consistent with \textit{Kerckhoffs's principle}, which states that everything
about the system should be known to an attacker \textit{except} for the secret key.
The objective of an attacker is to extract the secret key \begin{math} \mathtt{K} \end{math} from the hardware implementation of a given PSLL technique which would enable her to pirate the design IP and/or engage in overproduction of ICs.

Using the aforementioned resources and capabilities available to an attacker, we define two attack settings.

\begin{itemize}[leftmargin=*]

\item \textbf{Oracle-less setting:} An attacker $\mathcal{A}^{\mathbb{OL}}$ only uses the locked circuit \begin{math} \mathcal{C}_{lock} \end{math} to recover the secret key.

\item \textbf{Oracle-guided setting:} An attacker $\mathcal{A}^{\mathbb{OG}}$ uses the locked circuit \begin{math} \mathcal{C}_{lock} \end{math} and the oracle \begin{math} \mathbb{C}_{act} \end{math} to recover the secret key.

\end{itemize}

\subsection{Classification of PSLL Techniques}
\label{sec:PSLL}

A crypto-system exhibits \textit{provable security} when mathematical proofs exist showcasing resilience to certain attacks~\cite{buchmann2004introduction}.
A logic locking technique exhibits provable security when it is algorithmically secure against I/O-based attacks under the aforementioned threat model and assumptions discussed next.
\begin{itemize}[leftmargin=*]

\item The effort required by an attacker to determine the correct key $\mathtt{K}$, is exponential in the key-size $|K|$, \textit{i.e.,} $\mathcal{O}(2^{|K|})$.

\item An attacker is restricted from probing the oracle.

\end{itemize}

Based on the underlying hardware construction, we categorize PSLL techniques into \textit{hard-coded} and \textit{non-hard-coded} techniques.\footnote{Out of the many monikers used for different PSLL techniques, we adopt a simpler and generalized categorization of PSLL techniques.}
A hard-coded PSLL technique $\mathcal{L}^{\mathbb{HC}}$ constitutes a pair of algorithms (\textit{Perturb, Restore}).
The \textit{Perturb} algorithm takes the circuit \begin{math} \mathcal{C}_{orig} \end{math} and protected pattern $\mathtt{PP}$ as inputs and returns a functionality-stripped (or modified) circuit \begin{math} \mathcal{C}_{mod} \end{math}, and the associated key \begin{math} \mathtt{K} \end{math}.
The \textit{Restore} algorithm takes the modified circuit \begin{math}\mathcal{C}_{mod}\end{math} and augments a key-controlled restore unit \begin{math} \mathcal{C}_{restore} \end{math}, thereby generating a locked circuit \begin{math} \mathcal{C}_{lock} \end{math} (\begin{math} \mathcal{C}_{lock} \end{math} = \begin{math} \mathcal{C}_{mod} \end{math} $\oplus$ \begin{math} \mathcal{C}_{restore} \end{math}).
Conversely, a non-hard-coded PSLL technique $\mathcal{L}^{\mathbb{NHC}}$ consists just of algorithm \textit{Restore}.
The \textit{Restore} algorithm takes the circuit \begin{math} \mathcal{C}_{orig} \end{math}
as an input, augments a key-controlled restore unit \begin{math} \mathcal{C}_{restore} \end{math}, and returns a locked circuit \begin{math} \mathcal{C}_{lock} \end{math}, and the associated key \begin{math} \mathtt{K} \end{math}; \begin{math} \mathcal{C}_{lock} \end{math} = \begin{math} \mathcal{C}_{orig} \end{math} $\oplus$ \begin{math} \mathcal{C}_{restore} \end{math}.

Hard-coded PSLL techniques (Fig.~\ref{fig:PSLL_categories}(a)) comprise of techniques where the secret is hard-coded (the key is either hard-coded directly through $KIs$ or indirectly through $PIs$).
Examples include SARLock~\cite{SARLock_host_2016}, SFLL-HD$^{0}$~\cite{yasin_CCS_2017}, SFLL-flex~\cite{yasin_CCS_2017}, SFLL-rem~\cite{sengupta2020truly}, CAC~\cite{CAC}, SARLock-DTL~\cite{CAC}, CAC-DTL~\cite{CAC}, and error-controlled encryption (ECE)~\cite{shen2018comparative}.
A designer hard-codes the secret either by (i)~augmenting hard-coded point-functions~\cite{yasin_CCS_2017}, (ii)~replacing a few logic gates in point-functions with OR/NOR gates~\cite{CAC}, or (iii)~by removing logic~\cite{sengupta2020truly}.
On the other hand, non-hard-coded techniques (Fig.~\ref{fig:PSLL_categories}(b)) do not hard-code the secret key in the circuit.
Examples include Anti-SAT~\cite{xie16_SAT}, Anti-SAT-DTL~\cite{CAC}, CASLock~\cite{caslock_ches_2019}, Strong Anti-SAT (SAS)~\cite{liu2020strong}, and variants of Gen-Anti-SAT~\cite{zhou2021generalized}.
The construction comprises two logic functions, $f$ and $g$ (Fig.~\ref{fig:PSLL_categories}(b)), appended to the original circuit through an XOR gate.
While $f$ corresponds to $\overline{g}$ for Anti-SAT, Anti-SAT-DTL, CASLock, and Gen-Anti-SAT (Comp.), $f$ can be any function for Gen-Anti-SAT (Non-comp.)~\cite{zhou2021generalized}.

\begin{figure}[tb]
    \centering
    \includegraphics[width=0.9\columnwidth]{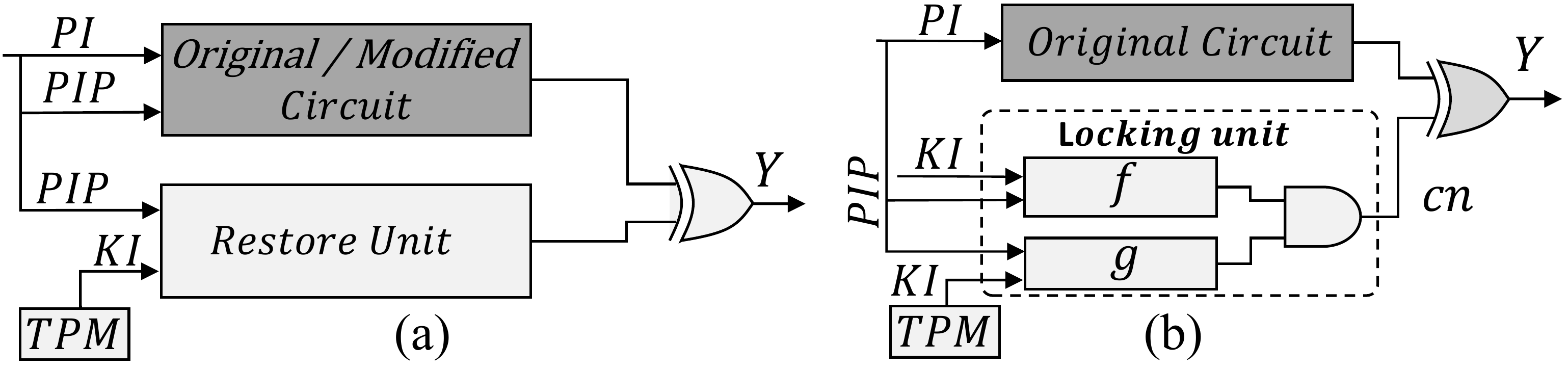}
    \caption{High-level construction of (a)~hard-coded PSLL techniques (b)~non-hard-coded PSLL techniques.
    $PI$ is primary input, $PIP$ is protected input port, $KI$ is key-input, \textit{TPM} is tamper-proof memory, and $cn$ is critical wire.
    For hard-coded PSLL techniques like SARLock~\cite{SARLock_host_2016} and ECE~\cite{shen2018comparative}, the original circuit is present (as is) and the secret is hard-coded in the restore unit through $KIs$.
    For techniques like CAC, CAC-DTL~\cite{CAC}, and SFLL variants~\cite{yasin_CCS_2017, sengupta2020truly}, the construction comprises a modified circuit and a restore unit.}
    \label{fig:PSLL_categories}
\end{figure}

\subsection{Primer on IC Testing}
\label{sec:primer_testing}

IC testing is a critical step in the supply chain that ensures that a fabricated chip does not possess manufacturing defects. 
A single stuck-at fault model is widely used to test circuits for faults~\cite{bushnell2004essentials}.
A wire is tested for both stuck-at-0 (s-a-0) and stuck-at-1 (s-a-1) fault.
Detecting an s-a-0 fault (at a wire) entails generating a test pattern ($\mathtt{TP}$) that sets the wire to logic 1.
There are three steps involved in generating $\mathtt{TP}$, viz., (i)~fault activation, (ii)~path sensitization, and (iii)~line justification.
Consider Fig.~\ref{fig:testing_princ}(a); we wish to detect an s-a-0 fault on $n1$.
$n1$ is output of an AND gate; hence, the input pattern activating $n1$ to logic 1 is $\langle a,b\rangle = \langle1,1\rangle$.
The next step is identifying a path to sensitize this value to a $PO$ ($O1$).
Since only one gate exists in the fan-out of $n1$, the path includes $n1$$\rightarrow$$O1$.
The final step is line justification. 
This step sensitizes other inputs of the logic gate connected to $n1$ to a known value.
As $n1$ fan-outs to NOR gate, the other input of NOR gate ($n2$) must be 0 to propagate the value at $n1$ to $O1$.
$n2$ is connected to an OR gate and is set to logic 0 using input pattern $\langle c,d\rangle = \langle0,0\rangle$.
Thus, the pattern generated to detect a s-a-0 fault at $n1$ is $\langle a,b,c,d\rangle = \langle 1,1,0,0\rangle$, as shown in Fig.~\ref{fig:testing_princ}(b). 
$\mathcal{T}$ generates $\mathtt{TP}$ to detect
faults at any given wire.

\begin{figure}[tb]
\centering
\includegraphics[width=0.95\textwidth]{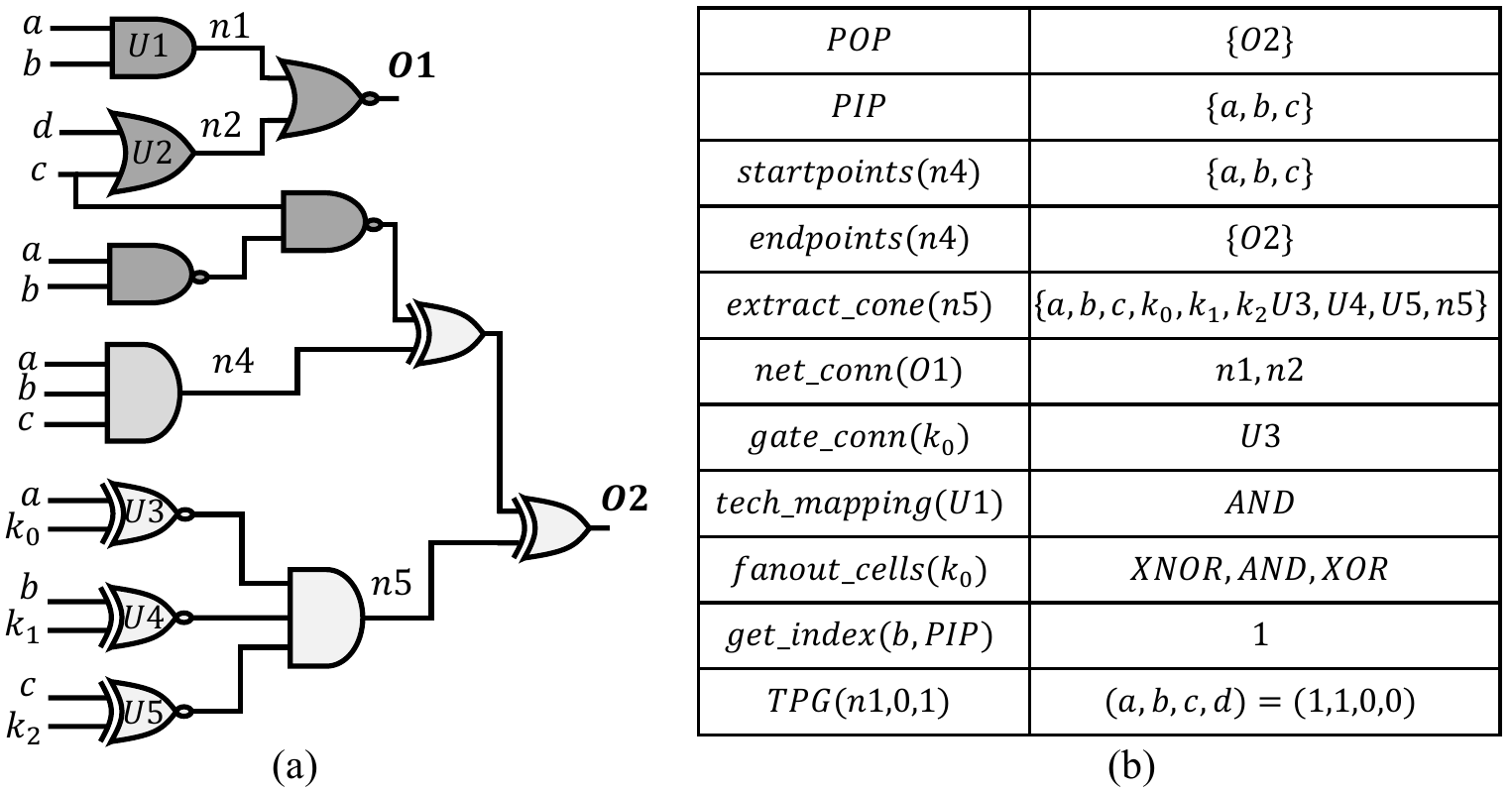}
\caption{Example to illustrate testing principles and common functions.}
\label{fig:testing_princ}
\end{figure}

\subsection{Common Functions Used in Our Key-Recovery Attacks}
\label{sec:functions}

Here we define the functions used in our key-recovery attacks.
The $PIs$ of a wire or gate $n_0$ is extracted using $\mathit{startpoints(n_0)}$; similarly $POs$ of net or gate $n_0$ is extracted using $\mathit{endpoints(n_0)}$.
A logic cone corresponding to a wire or gate $n_i$ is extracted using $\mathit{extract\_cone(n_i)}$.
A topologically sorted list of wires lying in the logic cone of $PO$ ($O_1$) is extracted using $\mathit{net\_conn(O_1)}$.
Gate connected to $KI$ ($K_0$) is obtained using $\mathit{gate\_conn(K_0)}$.
The type of logic gate $v_0$ is obtained using $\mathit{tech\_mapping(v_0)}$.
A topologically sorted list of gates in the fanout of input $K_0$ is obtained using $\mathit{fanout\_cells(K_0)}$.
An element at index $i$ in set $KI$ is denoted as $KI[i]$.
An index of element $k_i$ in set $KI$ is obtained using $\mathit{get\_index(k_i,KI)}$.
A set of test patterns $\mathtt{TP_i}$, where $i$ $\in$ \{0, $\dots$, d-1\} is obtained using an TPG tool ($\mathcal{T}$) to detect stuck-at-fault $\mathtt{f} \in \{0,1\}$ at wire $n_0$.
{$\mathcal{T}$($n_0$,$\mathtt{f}$,$\mathtt{d}$) generates $\mathtt{d}$ $\mathtt{TPs}$ to detect fault $\mathtt{f}$ at wire $n_0$, where $\mathtt{d}$ is a user-defined parameter.}
A synthesis tool ($\mathcal{S}$) is used to perform Boolean transformation using standard logic gates (NAND,AND,OR,NOR,XOR,XNOR,INV).
An attacker ($\mathcal{A}^{\mathbb{IO}}$) having access to an oracle $\mathbb{C}_{act}$ can launch the SAT-based attack $\mathit{SAT()}$ to recover the secret key $\mathtt{K}$.
The functions are explained using an example in Fig.~\ref{fig:testing_princ}(b).
\section{Attack on Hard-coded PSLL Techniques}
\label{sec:hardcoded_PSLL}

In this section, we conceptualize and develop an attack that recovers the secret key from the hardware implementation of hard-coded PSLL techniques.

\textbf{Problem formulation.} Given access to \begin{math} \mathcal{C}_{lock} \end{math} locked using $\mathcal{L}^{\mathbb{HC}}$ and black-box access to a working chip \begin{math} \mathbb{C}_{act} \end{math}, recover the secret key $\mathtt{K}_{rec}$ such that \begin{math} \mathbb{C}_{act} (i) \end{math} = \begin{math} \mathcal{C}_{lock} (i, \mathtt{K}_{rec}) \end{math}, \begin{math} \forall i \in PI \end{math}.

\subsection{Challenges}
\label{sec:hard_coded_RCs}

Our key-recovery attack aims to recover the secret hard-coded protected pattern ($\mathtt{PP}$) that induces output corruption in hard-coded PSLL techniques.
Protecting a circuit using a hard-coded PSLL technique becomes ineffective if the hard-coded $\mathtt{PP}$ does not influence output corruption.
Once the designer converts an algorithmic description of a PSLL technique into its equivalent hardware implementation, the underlying locked circuit becomes a sea of gates and wires.
From an attacker's perspective, examining each logic gate for leaking the secret key can be computationally challenging.
Moreover, the complexity is further exacerbated when designers utilize intricate logic optimization algorithms to generate locked circuits.
Additionally, hard-coded PSLL techniques use varied methods to hard-code the secret (\S\ref{sec:PSLL}).
Furthermore, some hard-coded PSLL techniques protect multiple $\mathtt{PPs}$~\cite{yasin_CCS_2017}.
Thus, an attacker faces the following challenges.

\begin{itemize}
\item [\textbf{C1}] How to identify potential vulnerabilities in a locked circuit and subsequently recover the secret key?
\item [\textbf{C2}] How to develop a generic key-recovery attack to challenge the security of hard-coded PSLL techniques?
In other words, the attack should be agnostic to the underlying construction of the hard-coded PSLL technique.
\end{itemize}

\subsection{Methodology}
\label{sec:methodology}

To address \textbf{C1}, we articulate the following properties rooted in the construction of hard-coded PSLL techniques.

\textbf{Property 1.} The locked circuit must remain testable, \textit{i.e.,} at least one input pattern exists that detects s-a-$0$ and s-a-$1$ faults at every net in the locked circuit.
Formally, \texttt{Pr}$[\mathcal{T}(n_i,f,\mathtt{d})$ = $\perp]$ = $0$, $\exists_\mathtt{d}$ $\forall_{i,f}$:  $i\in E$, $f\in\{0,1\}$.
This property is directly associated with the principles of IC testing~\cite{bushnell2004essentials} and applies to the hardware implementation of all PSLL techniques.

Recall that in a hard-coded PSLL technique, \textit{Perturb} algorithm generates $\mathcal{C}_{mod}$ from $\mathcal{C}_{orig}$, which can be accomplished either by (i)~inserting a hard-coded point-function or (ii)~removing functionality corresponding to this $\mathtt{PP}$.
$\mathcal{C}_{mod}$ will differ (in functionality) from $\mathcal{C}_{orig}$ only for $\mathtt{PPs}$, \textit{i.e.,} there could exist logic cone(s) inside $\mathcal{C}_{mod}$ which on the application of $\mathtt{PPs}$ as inputs would invert the output response of $\mathcal{C}_{orig}$ to induce corruption and obtain $\mathcal{C}_{mod}$.
We define the output of such logic cone(s) as \textit{key-revealing logic gate(s)}---these only activate on the application of $\mathtt{PPs}$ to induce output corruption.
Based on the construction of the hard-coded PSLL techniques considered in this work, we outline two properties that aid in identifying the potential key-revealing logic gate(s).

\textbf{Property 2.} Key-revealing logic gate(s) ($v \in V$) must be connected to $PIP$, \textit{i.e.,} they must have primary inputs $ins$ $\subseteq$ $PIP$ as the startpoints. 
Formally, $ins = \mathit{startpoints(v)}, v \in V$, $ins \subseteq PIP$.

Key-revealing logic gate(s) can be either connected to exactly the same number of $PIP$ (as the key-size) or a number lesser than the key-size to account for synthesis-induced transformations and merging with the original design.
Following the construction of hard-coded PSLL techniques, $PIP$ must be involved in the construction of the key-controlled restore unit.
For example, for techniques where $\mathtt{PP}$ is hard-coded (e.g., SFLL-HD$^0$), 
$PIP$ $\subseteq$ $PI$, whereas, for techniques where the $key$ is hard-coded (e.g., SARLock), $PIP$ $\subseteq$ $KI$.

\textbf{Property 3.} Key-revealing logic gate(s) ($v \in V$) must influence the corruption of only $POP$. 
Simply put, key-revealing logic gate(s) must have only $POP$ as the endpoints.
Formally, $POP = \mathit{endpoints(v)}, v \in V$.

Note that both properties (i)~are derived naturally from the hardware implementation of hard-coded PSLL techniques, and (ii)~prune the search space for an attacker.
However, with increased design complexity, \textit{i.e.,} designs with larger gate count (e.g., b19\_C with 237,962 gates), attackers might end up with a large number of key-revealing logic gate(s). 
To prune the search space further, we outline an additional property.

Recall that augmenting point functions (with hard-coded $\mathtt{PP}$) to the original circuit ensure I/O-based attacks prune out exactly one incorrect key in every attack iteration, thereby leading to stronger resilience (\S\ref{sec:attacks_defenses_PSLL}).
Therefore, the construction of hard-coded PSLL techniques necessitates the activation of key-revealing logic gate(s) only when $\mathtt{PP}$ is applied as an input pattern.
For instance, if the key-revealing logic gate outputs 0 for any non-$\mathtt{PP}$, then applying $\mathtt{PP}$ must toggle this value to $1$.
From an attacker's perspective, the challenge lies in recovering the secret $\mathtt{PP}$.
To recover this secret, an attacker can perform functional simulations (using random input patterns) and observe the output at the key-revealing logic gate(s).
However, obtaining input pattern(s) that justifies key-revealing logic gate(s) to $1$ through random simulations is challenging due to exponential complexity regarding the number of inputs.
\textit{This challenge is addressed by utilizing the principles of test pattern generation.}
More specifically, we use the stuck-at fault model to generate $\mathtt{TPs}$ that detect faults for any logic gate.
Finally, although most hard-coded PSLL techniques protect one $\mathtt{PP}$, some techniques protect multiple $\mathtt{PPs}$~\cite{yasin_CCS_2017}.
These points lead to the articulation of our next property.

\textbf{Property 4.} Key-revealing logic gate(s) ($v \in V$) must be activated for exactly $\mathtt{N}$ input patterns.
In other words, exactly $\mathtt{N}$ input patterns should set the value of the key-revealing logic gate(s) to 1.
Formally, 
\texttt{Pr}$[|\mathcal{T}(v,0,\mathtt{d})|$ = $\mathtt{N}]$ = $1$, $\mathtt{d} > \mathtt{N}$, $v \in V$.
$\mathtt{N}$ is a user-configurable parameter and depends on the underlying construction of the hard-coded PSLL technique. 
For instance, $\mathtt{N}$ is 1 for SFLL-HD$^{0}$ while it can be any number for SFLL-flex, which protects multiple $\mathtt{PPs}$.

\noindent\textbf{Generality.} The aforementioned properties are dictated by the construction of the considered hard-coded PSLL techniques.
Nevertheless, structural and functional properties corresponding to new PSLL techniques can be readily augmented to these properties, which ensures attack upgradability.
\textit{Note that these properties are applicable as long as the underlying PSLL technique demonstrates exponential complexity against I/O-based attacks.}
All our properties hold irrespective of how the defender constructs a modified circuit, \textit{i.e.,} either by adding logic (e.g., CAC~\cite{CAC}) or by removing logic (SFLL-rem~\cite{sengupta2020truly}). 

\begin{theorem}
\label{theorem_hard_coded_NEW}
A partial (or complete) secret can be recovered with test patterns using the stuck-at fault model for a hard-coded PSLL technique satisfying exponential SAT complexity.
\end{theorem}

\begin{proof}
Consider a $\mathtt{PP}$ is hard-coded in the original circuit ($\mathcal{C}_{orig}$) to obtain a modified circuit ($\mathcal{C}_{mod}$) such that the functionality of $\mathcal{C}_{mod}$ differs from $\mathcal{C}_{orig}$ only for $\mathtt{PP}$.
Irrespective of how $\mathcal{C}_{mod}$ is constructed to achieve exponential complexity (either by augmenting a point-function, diversifying point-function, or removing logic), it will constitute key-revealing logic gate(s) that will (i)~influence corruption of $PO$(s), (ii)~be testable for both faults, and (iii)~be activated only for $\mathtt{PP}$. 
As the effect of key-revealing logic gate(s) can be observed at $PO$, a stuck-at fault model can be used to generate a $\mathtt{TP}$ that controls the output of the key-revealing logic gate(s).
Considering Equation~\ref{eq:SL}, a $\mathtt{TP}$ is generated to test stuck-at fault $\mathtt{f}$, where $\mathtt{f}$ $\in$ \{0,1\} at the output of key-revealing logic gate(s), $cn$, in $\mathcal{C}_{mod}$.
For the correct fault, the computed $\mathtt{TP}$ contains partial (or complete) traces of the hard-coded secret.
\begin{align}\label{eq:SL}
	\mathtt{TP} \gets \mathcal{T}(cn,f,1); \hspace{10mm} \mathtt{TP} \in \mathtt{K}
\end{align}
Thus, $\mathtt{TP}$ generated using the stuck-at fault model can recover (either complete or partial) traces of the hard-coded secret.
\end{proof}

\begin{figure*}[htb]
    \centering
    \includegraphics[width=0.925\textwidth]{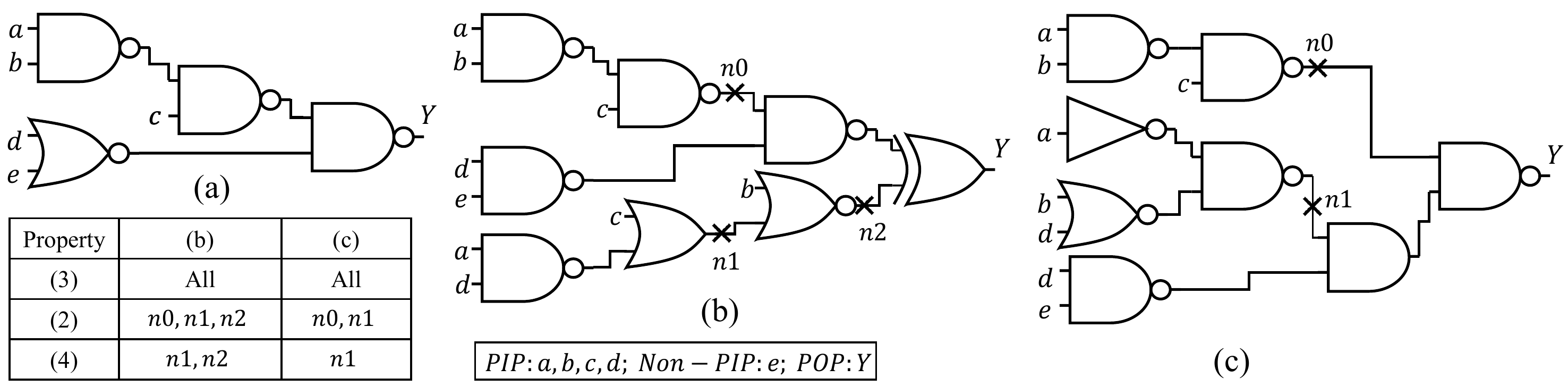}
    \caption{Example showing identification of key-revealing logic gate(s) using properties of hard-coded PSLL techniques.
    (a)~Original circuit 
    (b)~Modified circuit (with regards to functionality) for 
    $\mathtt{PP} = \langle1,0,0,1\rangle$ for $PIP$ $\langle a,b,c,d\rangle$.
    (c)~Modified circuit for 
    $\mathtt{PP} = \langle0,0,0,0\rangle$ for $PIP$ $\langle a,b,c,d\rangle$.
    Table outlines the key-revealing logic gate(s) identified per property. 
    For (b), there are two key-revealing logic gate(s) and for (c), there is only one key-revealing logic gate. 
    }
    \label{fig:KR_HC_eg}
\end{figure*}

\noindent\textbf{Idea.} As discussed, key-revealing logic gate(s) shall induce output corruption at $POP$ on the application of $\mathtt{PPs}$. 
Therefore, the effect of $\mathtt{PPs}$ can also be observed on internal nets of the key-revealing logic gate(s).
Such nets shall remain dormant (inactive) for non-$\mathtt{PPs}$ and activate only for $\mathtt{PPs}$.
Hence, an attacker can resort to generating a $\mathtt{TP}$ to detect s-a-0 (s-a-1) on nets in the key-revealing logic gate(s).
Such a $\mathtt{TP}$ will include traces of the hard-coded $\mathtt{PP}$.

\noindent\textbf{Example 1.} Consider $\mathcal{C}_{orig}$ in Fig.~\ref{fig:KR_HC_eg}(a) and $\mathcal{C}_{mod}$ in Fig.~\ref{fig:KR_HC_eg}(b) that is modified for the input pattern $\langle a,b,c,d\rangle = \langle1,0,0,1\rangle$.
Using property 3, all the nets are classified as key-revealing logic gate(s) since 
they affect the $POP$ ($Y$).
Using property 2, we perform a pruning operation on key-revealing logic gate(s).
This returns $n0$, $n1$, and $n2$ since they are connected to more than $\mathtt{N}/2$ $PIPs$, where $\mathtt{N}$ is the length of the $\mathtt{PP}$.
Using property 4, we further prune key-revealing logic gate(s) to $n1$ and $n2$.
$\mathcal{T}$ generates $\langle a,b,c,d,e\rangle = \langle1,\texttt{x},0,1,\texttt{x}\rangle$ for net $n1$ and $\langle a,b,c,d,e\rangle = \langle1,0,0,1,\texttt{x}\rangle$ for net $n2$.
$\mathtt{x}$ denotes undeciphered bit.
Thus, we recover the hard-coded secret $\langle a*,b*,c*,d*\rangle = \langle1,0,0,1\rangle$ in an oracle-less setting by testing for s-a-0 fault at $n2$.

\noindent\textbf{Example 2.} Consider $\mathcal{C}_{mod}$ in Fig.~\ref{fig:KR_HC_eg}(c) that is modified for the input pattern $\langle a,b,c,d\rangle = \langle0,0,0,0\rangle$.
All the nets are classified as key-revealing logic gate(s) using property 3.
Performing a pruning operation using property 2 returns $n0$ and $n1$. 
Using property 4, we are only left with $n1$.
Invoking $\mathcal{T}$ for $n1$ generates a $\mathtt{TP}$ as $\langle0,0,\mathtt{x},0,\mathtt{x}\rangle$.
Thus, we extracted a partial hard-coded secret key value ($\mathtt{K}$) as $\langle0,0,\mathtt{x},0\rangle$ by testing for stuck-at faults at $n1$.
The undeciphered bit ($\mathtt{x}$) can be revealed by querying an oracle through an I/O-based attack (e.g., SAT-based attack~\cite{subramanyan15}).
This example illustrates the scenario where an oracle is required to recover the secret key.

\subsection{Algorithm} 
\label{sec:algorithm_hard_coded_PSLL}

\begin{algorithm}[tb]
\footnotesize
\SetKwInput{KwInput}{Input}               
\SetKwInput{KwOutput}{Output} 
\DontPrintSemicolon
  
  \KwInput{Locked circuit ($\mathcal{C}_{lock}$), Oracle ($\mathbb{C}_{act}$), Primary inputs ($PI$), Key-inputs ($KI$), Number of protected patterns ($|\mathtt{PP}|$)}
  \KwOutput{Secret key ($\mathtt{K}$)}

\SetKwFunction{FKR}{KeyRecovery\_HC}
\SetKwFunction{FPM}{KeyInputMapping}
\SetKwFunction{FKE}{KeyExtraction}
\SetKwFunction{FSA}{ExtractNets}
\SetKwFunction{FCN}{CandidateNets}

\SetKwProg{Fn}{procedure}{:}{}
    \Fn{\FSA{keyin,in}}{
        $POP$ $\gets$ $\mathit{endpoints(keyin)}$\;
        $N$ $\gets$ $\mathit{net\_conn(POP)}$\; 
        \KwRet $N$\;
    }
  
\SetKwProg{Fn}{procedure}{:}{}
    \Fn{\FPM{N}}{
        $PIP$,$KEY$ $\gets$ $\emptyset$ \;
		\For{$net$ $\in$ $N$}{
			$ins$ $\gets$ $\mathit{startpoints(net)}$\;
			\If{$|ins|$ = $2$}{
				($key$,$pip$) $\gets$ $ins$ where $key$ $\in$ $KI$\; 
				$PIP$.$\mathit{append(pip)}$\;
				$KEY$.$\mathit{append(key)}$\; 
			}
		}
		\KwRet $PIP$, $KEY$\;
    }  
	
\SetKwProg{Fn}{procedure}{:}{}
    \Fn{\FCN{N,PIP}}{
        $CN$ $\gets$ $\emptyset$ \;
		\For{$net$ $\in$ $N$}{
			$in$ $\gets$ $\mathit{startpoints(net)}$\;
			\If{$in$ $\in$ $PIP$} {
				$CN$.$\mathit{append(net)}$\;
			}
		}	
        \KwRet $CN$\;
    }  

\SetKwProg{Fn}{procedure}{:}{}
	\Fn{\FKE{$\mathtt{TP}$,$PIP$,$\mathtt{K}$,$PI$}}{
		\For{$k$ in \{0,$|\mathtt{K}|$\}} {
			$pip$ $\gets$ $PIP$[$k$]\;
			$idx$ $\gets$ $\mathit{get\_index(pip,PI)}$\;
            $val$ $\gets$ $\mathtt{TP}$[$idx$]\;
			\If{$val$ $\neq$ $\mathtt{x}$} {
				$key$[$k$] $\gets$ $val$\;
			}	
		}
		\KwRet $key$\;
	}

\SetKwProg{Fn}{function}{:}{\KwRet}
	\Fn{\FKR}{
        \{${N}$\} $\gets$ \texttt{ExtractNets($KI$,$PI$)}\;
		\{${KEY}$,${PIP}$\} $\gets$ \texttt{KeyInputMapping($N$)}\;
        \{$CN$\} $\gets$ \texttt{CandidateNets($N$,$PIP$)}\;
        \For{$cn$ in $CN$ }{
            $\mathtt{TP}$ $\gets$ $\mathcal{T}$($cn$,$i$,$|\mathtt{PP}|$+1); $i$ $\in$ \{0,1\}\;
            \If{$|\mathtt{TP}|$ $=$ $|\mathtt{PP}|$}{
				 $\mathtt{K}$ $\gets$ \texttt{KeyExtraction($\mathtt{TP}$,$PIP$,$KEY$,$PI$)}\;
                \If{$|\mathtt{K}|$ $<$ $|KI|$}{
                    $key\_final$ $\gets$ $\mathit{SAT(\mathcal{C}_{lock},\mathbb{C}_{act},\mathtt{K})}$\;
                    $\mathtt{K}$ $\gets$ $key\_final$\;
                }
                \KwRet $\mathtt{K}$\;
            }
        }
	}
\caption{{Attack on hard-coded PSLL techniques}
\label{alg:KR_HC}}
\end{algorithm}

We outline our attack on hard-coded PSLL techniques in Alg.~\ref{alg:KR_HC}.
It consists of five steps, viz., (i)~extraction of nets in the $POP$, (ii)~obtaining mapping between $KIs$ and $PIPs$, (iii)~identification of key-revealing logic gate(s), (iv)~generation of $\mathtt{TPs}$, and (v)~recovery of the secret key.
\texttt{ExtractNets()} returns the list of nets lying in the $POP$ (lines 1--4).
\texttt{KeyInputMapping()} extracts the mapping between $KIs$ and $PIPs$ from the key-controlled restore unit (lines 5--13).
\texttt{CandidateNets()} returns the
key-revealing logic gate(s) that satisfy the 
properties discussed in \S\ref{sec:methodology} (lines 14--20).
Next, the algorithm utilizes $\mathtt{TPs}$ to recover the hard-coded $\mathtt{PP}$.
$\mathcal{T}$ generates $\mathtt{TPs}$ that detects s-a-0 and s-a-1 for key-revealing logic gate(s).
Given a PSLL technique, $\mathcal{T}$ generates $\mathtt{d}$ $\mathtt{TPs}$. 
For example, when considering SFLL-HD$^0$, $\mathcal{T}$ is queried to generate \textit{exactly} two $\mathtt{TPs}$ for key-revealing logic gate(s).
The correct key-revealing logic gate(s) will return \textit{only} one $\mathtt{TP}$ and any key-revealing logic gate(s) generating more than one $\mathtt{TP}$ is(are) discarded.
$\mathtt{TP}$ is fed to $\mathtt{KeyExtraction()}$ to extract key-bits corresponding to $PIPs$ (lines 21--28).
If all key-bits are recovered from $\mathtt{TPs}$, the algorithm outputs it as the secret key.
However, there might be scenarios where partial key-bits are recovered from multiple $\mathtt{TPs}$.
In such cases, partial key-bits can be combined from the multiple $\mathtt{TP}$.
Furthermore, if $\mathtt{TPs}$ cannot recover some key-bits, $\mathit{SAT()}$ is invoked to recover them.
Finally, the algorithm merges the partial key recovered from $\mathtt{TPs}$ and the key returned from the $\mathit{SAT()}$ to output the final key.

\begin{figure}
    \centering
    \includegraphics[width=0.9\textwidth]{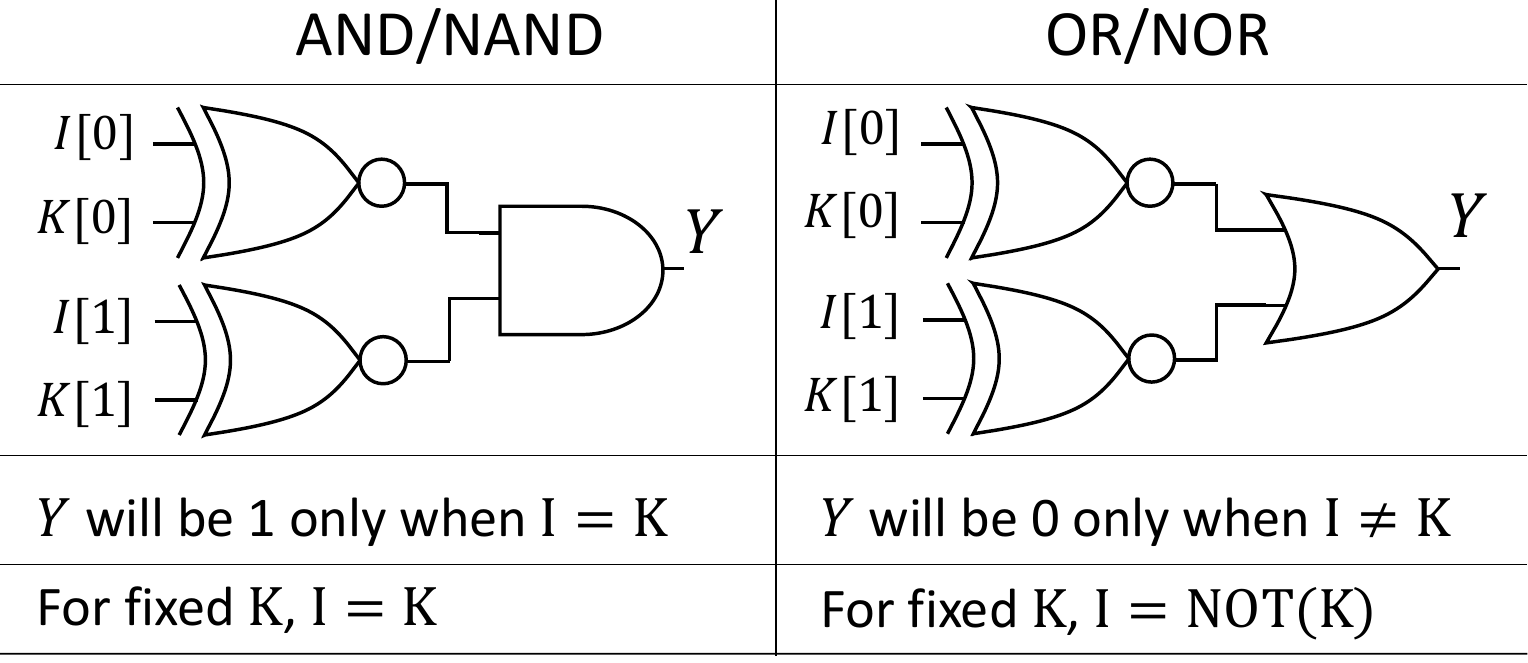}
    \caption{Key-recovery example where point-functions are diversified.}
    \label{fig:KR_HC_CAC_DTL}
\end{figure}

\subsection{Extension to Other Hard-coded PSLL Techniques}
\label{sec:extenstion_hard_coded_PSLL}

Recall that \textbf{C2} outlines the challenge of developing a generic key-recovery attack agnostic to the construction of a hard-coded PSLL technique (\S\ref{sec:hard_coded_RCs}).
Next, we discuss the (minor) modifications we implemented to address \textbf{C2}.

Some hard-coded PSLL techniques (e.g., CAC-DTL~\cite{CAC}, and variants of SFLL~\cite{yasin_CCS_2017,sengupta2020truly}) hard-code the secret through $PIs$ in the original circuit, while others (e.g., SARLock-DTL~\cite{CAC}, ECE~\cite{shen2018comparative}) hard-code the secret through $KIs$ in the restore unit.
Our attack successfully recovers the secret key across both classes of techniques.
Since the $\mathtt{PP}$ corresponds to $PIPs$ for techniques such as CAC, CAC-DTL, and SFLL variants, the attack extracts key-bit corresponding to $PIPs$ from $\mathtt{TPs}$.
For SARLock, SARLock-DTL, and ECE, the $\mathtt{PP}$ corresponds to $KIs$; hence, instead of $PIPs$, key-bits corresponding to $KIs$ are extracted from $\mathtt{TP}$.
This is achieved by removing lines 22--23 in Alg. 1 and modifying line 24 to $val$ $\gets$ $\mathtt{TP}$[$k$].

The second scenario we address towards the generality of our key-recovery attack is to account for the number of $\mathtt{PPs}$.
While SFLL-HD$^0$ protects exactly one $\mathtt{PP}$, SFLL-flex protects multiple $\mathtt{PPs}$.
Our attack addresses this using property 4 and utilizes a user-configurable parameter \begin{math} \mathtt{TP} \end{math} that constrains the $\mathcal{T}$ tool to produce exactly $\mathtt{N}$ test patterns.

Finally, in scenarios where a designer diversifies the hard-coded point function using OR/NOR gates~\cite{CAC}, we recover the secret key as follows.
When a designer replaces/diversifies some AND gates in hard-coded point-function with NAND/OR/NOR gates, we slightly modify the key extraction strategy from $\mathtt{TPs}$.
Fig.~\ref{fig:KR_HC_CAC_DTL} illustrates the relation between $I$ and hard-coded key $K$ for two examples.
When AND gates are replaced with NAND gates, there is no change in $\mathtt{KeyExtraction()}$.
However, when AND gates are replaced with OR/NOR gates, the relation between $K$ and $I$ changes, as shown in column 2, row 4.
Thus, with modifications to $\mathtt{KeyExtraction()}$, our attack recovers the secret key for different versions of DTL.
\section{Attack on Non-hard-coded PSLL Techniques}
\label{sec:non_hardcoded_PSLL}

In this section, we develop an attack that recovers the secret key from the hardware implementation of non-hard-coded PSLL techniques.
The problem formulation is the same as mentioned in \S\ref{sec:hardcoded_PSLL} and is omitted here.

\subsection{Challenges}

The construction of non-hard-coded PSLL techniques consists of a key-controlled locking unit appended to the original circuit via one critical wire $cn$ (Fig.~\ref{fig:PSLL_categories}(b)).
It should be noted that the secret resides in the key-controlled locking unit, and therefore, the first step for an attacker is to structurally analyze the locked circuit to identify the critical wire $cn$ that separates the locking unit from the original circuit.
However, as discussed in \S\ref{sec:hard_coded_RCs}, the complexity of identifying this wire (net) is challenging and further exacerbated due to synthesis-guided logic optimizations.
The next challenge is how to recover the secret key from the locking unit and how can a generic attack be developed for non-hard-coded PSLL techniques, independent of the construction of the locking unit.
To summarize, an attacker faces the following challenges.

\begin{itemize}

\item [\textbf{C3}] How to identify the locking unit from a locked circuit and recover the secret key from the locking unit?

\item [\textbf{C4}] How to develop a generic key-recovery attack for non-hard-coded PSLL techniques?

\end{itemize}

\subsection{Methodology}
\label{sec:methodology_non_hard_coded}

To address \textbf{C3}, the first step
entails identifying the locking unit from the locked circuit.
We outline a property stemming from the construction of non-hard-coded PSLL techniques.

\textbf{Property 5.} The wire $cn$ (that separates the locking unit from the original circuit) must be connected to all $KI$ and all $PIP$.
Also, $cn$ must influence the output corruption of $POP$. 
Formally, $\{KI,PIP\} = \mathit{startpoints}(cn)$ and $POP = \mathit{endpoints}(cn)$.
If multiple candidates for $cn$ exist, then the wire closest (shortest distance measured in levels of logic) to the $KI$ is chosen for extracting the locking unit.

After successfully extracting the locked unit, the next step involves recovery of the secret key from the locking unit.
To that end, we first provide the definition of key-gate mapping.

Recall the construction of non-hard-coded PSLL techniques where a key-controlled locking unit is XORed with the original circuit to obtain a locked circuit (Fig.~\ref{fig:PSLL_categories}(b)).
For Anti-SAT, blocks $f$ and $g$ are complementary to each other and denoted by $g$ and $\overline{g}$.
The blocks are controlled by the same $PIPs$ but different $KIs$ ($K$ = \{$\mathcal{K}_1$, $\mathcal{K}_2$\}), where $\mathcal{K}_j$ consists of key-inputs $k_{ji}$; $i \in \{0,\dots,|K/2|-1\}$, $j \in \{1,2\}$.
Locking unit can be formally defined as $Y$ = $g(PIP_{i}\oplus k_{1i}\oplus r_{1i})$ $\wedge$ $\overline{g(PIP_{i}\oplus k_{2i}\oplus r_{2i})}$, $\forall i \in \{0,|K|/2-1\}$; $r$ as 0(1) indicates an X(N)OR key-gate.
This construction forces I/O-based attacks to query at least $2^{|K|/2}$ input patterns, where $|K|$ = $|\mathcal{K}_1|$ + $|\mathcal{K}_2|$ is the total key-size.
For instance, for a circuit locked using Anti-SAT with key-size $|K|$, it takes $2^{|K|/2}$ queries to recover the secret key; each query to an oracle eliminates $2^{|K|/2}-1$ incorrect keys.
This means, there are $2^{|K|/2}$ correct keys, and when analyzed further, we identify a unique mapping between $\mathcal{K}_1$ and $\mathcal{K}_2$ portions of the correct keys, \textit{i.e.,} $\mathcal{K}_1$ $\oplus$ $\mathcal{K}_2$ across all the correct keys is unique.
We define this unique mapping as \textit{key-gate mapping} (KGM) and seek to find this mapping to recover the secret key(s) from the hardware implementation of a non-hard-coded PSLL technique. 
Finally, we introduce our last property, which considers the gate-type for $r$ (X(N)OR) in the successful recovery of the unique mapping. 
Unlike the properties discussed so far, this is an assumed property.

\textbf{Property 6.} Each $KI$ must drive exactly one X(N)OR logic gate. Formally, \texttt{Pr}$[\mathit{gate\_conn(k_i)}$ = X(N)OR$]$ = $1$, $k_i \in KI$.

\begin{theorem}
\label{NHC_theorem}
A key-gate mapping exists between the two sets of keys in non-hard-coded PSLL techniques.
\end{theorem}

\begin{proof}
As per the aforementioned definition of locking unit, $Y$ depends on $PIPs$ and $KIs$ ($k_{1i}$ and $k_{2i}$) and will only be $0$ when ($PIP_i \oplus k_{1i}\oplus r_{1i}$) equals ($PIP_i \oplus k_{2i}\oplus r_{2i}$).
Operation $k_{ji}\oplus r_{ji}$ denotes either XOR or XNOR key-gate $\forall j \in \{1,2\}$ and the corresponding key-values can be either $\langle0,0\rangle$ or $\langle1,1\rangle$ when $r_{1i}$ $=$ $r_{2i}$, or $\langle0,1\rangle$ or $\langle1,0\rangle$ when $r_{1i}$ $\neq$ $r_{2i}$. 
Thus, there is a definite mapping between the key-gates in $\mathcal{K}_{1}$ and $\mathcal{K}_2$ bins, which can derive the secret key(s).
\end{proof}

In a pre-synthesized locked circuit, the construction of the non-hard-coded PSLL technique is retained ($g$ $\wedge$ $\overline{g}$).
Thus, the correlation between individual key-inputs can be derived.
For example, consider Fig.~\ref{fig:antisat_block}(a), 
$k0$ is correlated with $k4$ since both the $KIs$ are XORed with the same $PIP$ $I0$. 
Similarly $k2$ is correlated with $k6$ since they are both XORed with the same $PIP$ $I2$.
Subsequently, all the correlated $KIs$ are distributed into two bins, $\mathcal{K}_1$ and $\mathcal{K}_2$.
Once the correlation between $KIs$ is identified, and the subsequent binning process is completed, our KGM attack extracts the following attributes per $KI$.

\begin{itemize}
\item Bin it belongs to ($\mathcal{K}_1$ or $\mathcal{K}_2$)

\item Gate it is connected to (XOR or XNOR)?

\item Is there an inversion (or not) on its output path?

\end{itemize}

Table~\ref{tab:key_mapping} describes the recovery of the secret key using the aforementioned attributes.
Considering columns 2 and 3, keeping key-bin and gate-type constant, the effect of inversion can be observed on the key-value. 
Similarly, consider columns 6 and 8; the effect of gate-type can be observed on the key-value, while with columns 5 and 9, we can observe the effect of key-bin on the key-value. 
Next, we discuss the application of the KGM attack using two examples for Anti-SAT.

\begin{figure}[tb]
\centering
\includegraphics[width=\textwidth]{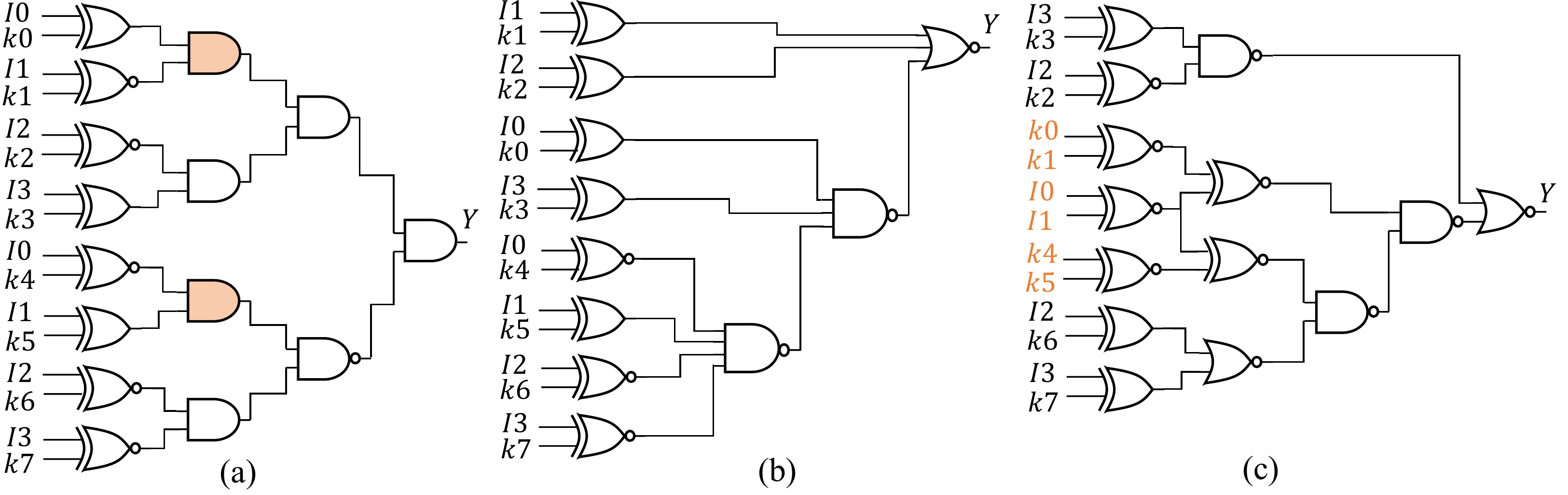}
\caption{a)~Pre-synthesis Anti-SAT block. (b)~Post-synthesis Anti-SAT block (all standard gates). (c)~Post-synthesized Anti-SAT-DTL block when gates in orange in (a) are replaced with XOR gates.
}
\label{fig:antisat_block}
\end{figure}

\noindent\textbf{Example 3.} Consider the pre-synthesized Anti-SAT locking unit in Fig.~\ref{fig:antisat_block}(a), where $k0$ and $k4$ are connected to an XOR and XNOR gate. 
As per the proof of Theorem~\ref{NHC_theorem}, key-values for these two key-bits should be either $\langle0,1\rangle$ or $\langle1,0\rangle$.
Next, we check the key-values obtained using our attack.
Before applying key-gate mapping, we distribute the $KIs$, $k0$ and $k4$, corresponding to the same $PI$ ($I0$), into two bins, $\mathcal{K}_1$ and $\mathcal{K}_2$.
Observe that $k4$ sees inversion on its signal path (belongs to $\overline{g}$), whereas $k0$ sees no inversion (belongs to $g$). 
Thus, as per Table~\ref{tab:key_mapping}, $k0$ is 0 and $k4$ is 1, consistent with our result above. 
Note that the key-mapping between the two sets must be unique and hence there are two correct values corresponding to $k0$ and $k4$, $\langle0,1\rangle$ and $\langle1,0\rangle$.
Similarly, we recover the remaining key-values, leading to the secret key $\langle k0,k1,k2,k3,k4,k5,k6,k7\rangle$ as $\langle0,1,1,0,1,0,1,1\rangle$ and the secret mapping between $\mathcal{K}_{1}$ and $\mathcal{K}_{2}$ ($\mathcal{K}_{1}$ $\oplus$ $\mathcal{K}_{2}$) as $\langle1,1,0,1\rangle$.
All the (2$^{4}$) keys satisfying this mapping are the secret keys.

\noindent\textbf{Example 4.} Consider the post-synthesized Anti-SAT locking unit in Fig.~\ref{fig:antisat_block}(b), where the $KIs$ are connected to XOR and XNOR gates.
First, we perform key-binning by distributing the $KIs$ into two bins.
Thus, \{$k0$, $k1$, $k2$, $k3$\} are binned into $\mathcal{K}_{1}$, while \{$k4$, $k5$, $k6$, $k7$\} are binned into $\mathcal{K}_{2}$. 
We utilize the strategy outlined in Table~\ref{tab:key_mapping} and traverse the signal path of the logic gates connected to each $KI$ to keep track of any signal inversions.
Note that NAND/NOR gates also account for inversions in addition to INV gates. 
We recover the secret key value $\mathtt{K}$ as $\langle0,1,1,0,1,0,1,1\rangle$ by scrutinizing the attributes of each $KI$ as per Table~\ref{tab:example3}.

\noindent\textbf{Corner cases.} To address \textbf{C4}, we identify various corner cases and address them in our attack.
Some synthesis-induced optimizations
may challenge an attacker in successfully identifying the correlation between $KIs$.
We can address this challenge by performing another round of logic optimization
on the extracted locking unit using only standard logic gates.
Once correlation between the $KIs$ is obtained and correlated $KIs$ are distributed into distinct bins, the value of key-bits can be recovered using Table~\ref{tab:key_mapping} in an oracle-less setting.

\noindent\textbf{Example 5.} Consider the post-synthesized circuit in Fig.~\ref{fig:antisat_block}(c), which is obtained when gates marked in orange in Fig.~\ref{fig:antisat_block}(a) are replaced with XOR gates (Anti-SAT-DTL) and subsequently synthesized. 
First, we perform key-binning by distributing the $KIs$ into two bins depending on their connected $PIPs$. 
Thus, \{$k0$, $k1$, $k2$, $k3$\} are binned into $\mathcal{K}_{1}$, while \{$k4$, $k5$, $k6$, $k7$\} are binned into $\mathcal{K}_{2}$.
However, \{$k0$, $k1$, $k4$, $k5$\} do not have a definite correlation.
We utilize the strategy outlined in Table~\ref{tab:key_mapping} for the correlated $KIs$.
Note that due to inconclusive correlation between key-inputs \{$k0$, $k1$, $k4$, $k5$\}, key-bits corresponding to $k0$ and $k1$ remain undeciphered using our KGM attack.
We assign $k4$ and $k5$ random key-bits since they lie in a different bin than $k0$ and $k1$.
We invoke the SAT-based attack to recover these two key-bits.
We successfully recover the secret key value $\mathtt{K}$ as $\langle1,1,1,0,0,0,1,1\rangle$ by scrutinizing the attributes of correlated $KIs$ as per Table~\ref{tab:example5} and invoking SAT-based attack for the undeciphered key-bits.

Our KGM attack recovers the key in an oracle-less setting if all $KIs$ are successfully correlated and binned into distinct sets.
However, assume $KIs$ are (i)~connected to the same XOR/XNOR gate, (ii)~connected to gates other than XOR/XNOR, or (iii)~driving multiple logic gates.
In such cases, binning is unsuccessful, and the attack annotates them as undeciphered key-bits ($\mathtt{x}$).
Our KGM attack recovers these key-bits by executing the SAT-based attack~\cite{subramanyan15}.

\begin{table}[tb]
\footnotesize
\caption{Key-gate mapping table. Key-bit value is dependent on key bin, gate-type, and inversion in the fanout}
\setlength{\tabcolsep}{0.9mm}
\label{tab:key_mapping}
\begin{tabular}{c|aacbdcdb}
& \multicolumn{8}{c}{\textbf{Key-inputs}}                                                              \\ \cline{1-9}
\textbf{Key-bin}												 	 & 1           & 1           & 1           & 1           & 2           & 2           & 2           & 2           \\     
\textbf{Gate-type}                                                   & XOR         & XOR         & XNOR         & XNOR         & XOR        & XOR         & XNOR        & XNOR        \\ 
\textbf{\begin{tabular}[c]{@{}c@{}}Inversion?\end{tabular}}          & No         & Yes         & No      & Yes         & No         & Yes       & No   & Yes         \\ 
\hline
\textbf{Key-value}                                                   & 0           & 1           & 1           & 0           & 1           & 0           & 0          & 1           \\ 
\hline
\end{tabular}
\end{table}

\begin{table}[tb]
\caption{Extracting attributes of each key-input and recovering secret key from locking unit for example 4.}
\label{tab:example3}
\setlength{\tabcolsep}{1.2mm}
\footnotesize
\begin{tabular}{c|cccccccc}
\textbf{Key-input}                                                   & \textbf{k0} & \textbf{k1} & \textbf{k2} & \textbf{k3} & \textbf{k4} & \textbf{k5} & \textbf{k6} & \textbf{k7} \\ \hline
\textbf{Key-bin}													 & 1           & 1           & 1           & 1           & 2           & 2           & 2           & 2           \\     
\textbf{Gate-type}                                                   & XOR         & XOR         & XOR         & XOR         & XNOR        & XOR         & XNOR        & XNOR        \\ 
\textbf{\begin{tabular}[c]{@{}c@{}}Inversion?\end{tabular}}          & No          & Yes         & Yes         & No          & Yes         & Yes         & Yes         & Yes         \\
\hline
\textbf{Key-value}                                                   & 0           & 1           & 1           & 0           & 1           & 0           & 1           & 1           \\ 
\hline
\end{tabular}
\end{table}

\begin{table}[t]
\caption{Extracting attributes of each key-input and recovering secret key from locking unit for example 5}
\label{tab:example5}
\setlength{\tabcolsep}{1.5mm}
\footnotesize
\begin{tabular}{c|cccccccc}
\textbf{Key-input}                                                   & \textbf{k0} & \textbf{k1} & \textbf{k2} & \textbf{k3} & \textbf{k4} & \textbf{k5} & \textbf{k6} & \textbf{k7} \\ \hline
\textbf{Key-bin}													 & 1           & 1           & 1           & 1           & 2           & 2           & 2           & 2           \\     
\textbf{Gate-type}                                                   & -         & -         & XOR         & XOR         & -        & -         & XNOR        & XNOR        \\ 
\textbf{\begin{tabular}[c]{@{}c@{}}Inversion?\end{tabular}}          & -          & -         & Yes         & No          & -         & -         & Yes         & Yes         \\
\hline
\textbf{Key-value}                                                   & $\mathtt{x}$           & $\mathtt{x}$           & 1           & 0           & 0           & 0           & 1           & 1           \\ 
\hline
\end{tabular}
\\[1mm]
- denotes unknown value
\end{table}

\subsection{Algorithm}
\label{sec:non_PSLL_algorithm}

We outline our generic attack on non-hard-coded PSLL techniques in Alg.~\ref{alg:KR_NHC}.
It consists of five steps, viz., (i)~identification of critical wire, (ii)~extraction of the locking unit, (iii)~re-synthesis of extracted logic cone, (iv)~extraction of attributes, and (v)~recovering key-bits.
Identification of the critical wire ($cn$) and subsequent extraction of the locking unit ($\mathcal{C}_{cn}$) takes place in \texttt{ExtractLogicCone()}.
The extracted logic cone, $\mathcal{C}_{cn}$, is re-synthesized using $\mathcal{S}$
with only the standard gates to generate $\mathcal{C}_{cn\_syn}$.
Next, \texttt{GetAttribute()} is executed on $\mathcal{C}_{cn\_syn}$ which extracts attributes such as the set each $KI$ belongs to, the type of logic gates connected to each $KI$, and the presence of inversion on the output path.
These attributes are passed to \texttt{KeyMapping()} to recover the secret key corresponding to the $KIs$.
Suppose the algorithm cannot conclusively determine the attribute of any $KI$. 
In that case, the value of that $KI$ is annotated as $\mathtt{x}$. 
Finally, we invoke $\mathit{SAT()}$ to recover the undeciphered $KIs$ using an oracle, $\mathcal{C}_{act}$.

\noindent\textbf{Time complexity.} Table~\ref{tab:time_complexity_paper} showcases the time complexity of functions and procedures used in our key-recovery attack algorithms (Alg.~\ref{alg:KR_HC} and~\ref{alg:KR_NHC}).
Although ATPG is an NP-complete problem, efficient heuristics have been developed for practical circuits, which reduces this complexity to polynomial in the number of gates in the circuits~\cite{prasad1999atpg}.
Synthesis approaches have undergone decades of research, with worst-case complexity being $\mathcal{O}(E^3)$~\cite{sasao1993logic}. 
In this work, since we use a commercial, closed-source synthesis tool (\textit{i.e.,} \textit{Synopsys DC}), the time complexity of our KGM attack algorithm cannot be conclusively obtained.
The overall time complexity of our attacks is contingent on the underlying structure (graph representation) and the number of gates in the locked circuit.

\begin{table}[ht]
\caption{Time complexity of the functions and procedures used in our key-recovery attack algorithms}
\footnotesize
\begin{tabular}{lc}
\hline
\textbf{Function}   & \textbf{Time complexity} \\ \hline \hline
$startpoints()$ & $\mathcal{O}(|V|+|E|)$ \\ \hline
$endpoints()$ & $\mathcal{O}(|V|+|E|)$ \\ \hline
$net\_conn()$ & $\mathcal{O}(|V|+|E|)$ \\ \hline
$get\_index()$ & $\mathcal{O}$(1) \\ \hline
$append()$ & $\mathcal{O}$(1) \\ \hline
$extract\_cone()$ & $\mathcal{O}(|V|+|E|)$ \\ \hline
$gate\_conn()$ & $\mathcal{O}$(1) \\ \hline
$tech\_mapping()$ & $\mathcal{O}$(1) \\ \hline
$fanout\_cells()$ & $\mathcal{O}(|V|+|E|)$ \\ \hline
$\texttt{ExtractNets()}$ & $\mathcal{O}(|V|+|E|)$ \\ \hline
$\texttt{KeyInputMapping()}$ & $\mathcal{O}(|V|+|E|)$ \\ \hline
$\texttt{CandidateNets()}$ & $\mathcal{O}(|N|*(|V|+|E|))$ \\ \hline
$\texttt{KeyExtraction()}$ & $\mathcal{O}(|K|)$ \\ \hline
$\texttt{ExtractLogicCone()}$ & $\mathcal{O}(|N|*(|V|+|E|))$ \\ \hline
$\texttt{GetAttribute()}$ & $\mathcal{O}(|V|+|E|)$ \\ \hline
$\texttt{KeyMapping()}$ & $\mathcal{O}$(1) \\ \hline
\end{tabular}
\label{tab:time_complexity_paper}
\\
[1ex]
V: vertices (gates) \hspace{5pt} E: edges (nets) \hspace{5pt}
N: candidate nets (N $\subset$ E)
\end{table}

\begin{algorithm}[ht]
\footnotesize
\SetKwInput{KwInput}{Input}                
\SetKwInput{KwOutput}{Output}
\DontPrintSemicolon
  
  \KwInput{Locked design ($\mathcal{C}_{lock}$), Oracle ($\mathbb{C}_{act}$), Key-inputs ($KI$)}
  \KwOutput{Secret key (\texttt{K})}

  \SetKwFunction{FEC}{ExtractLogicCone}
  \SetKwFunction{FGA}{GetAttribute}
  \SetKwFunction{FKM}{KeyMapping}
  \SetKwFunction{FKR}{KeyRecovery\_NHC}

\SetKwProg{Fn}{procedure}{:}{}
  \Fn{\FEC{c,keyin}}{
        $POP$ $\gets$ $\mathit{endpoints}(keyin)$\;
        $N$ $\gets$ $\mathit{net\_conn}(POP)$\;        
        \For{$net$ $\in$ $N$ }{
            $ins$ $\gets$ $\mathit{startpoints}(net)$\;
            \If{ ($keyin$ $\subseteq$ $ins$)} {
                $cn$ $\gets$ $net$\;
            }
        }
        $\mathcal{C}_{cn}$ $\gets$ $\mathit{extract\_cone}(\mathcal{C},cn)$\;
        \KwRet $\mathcal{C}_{cn}$\;
  }
  
\SetKwProg{Fn}{procedure}{:}{}
  \Fn{\FGA{key,PIP,gate}}{
		$bin$ $\gets$ $1$\;
		\{$ins$\} $\gets$ $\mathit{startpoints}(gate)$\;
		\If{$ins$ $\subseteq$ $PIP$} {
			$bin$ $\gets$ $2$\;	
			$PIP$.\texttt{append}($ins$)\; 
		} 
		$lib$ $\gets$ $\mathit{tech\_mapping}(gate)$\;
		$inv$ $\gets$ $0$\;
        \If{`INV' $\subseteq$ $\mathit{fanout\_cells(key)}$}{
            $inv$ $\gets$ $1$\;
        }
        \KwRet $bin$,$lib$,$inv$\;
  }  
  
\SetKwProg{Fn}{procedure}{:}{}
  \Fn{\FKM{bin, gateType, inv}}{
        $key\_value$ $\gets$ $0$\;
        \If{ ($bin$ = $1$ \& $inv$ = $1$ \& $lib$ = `XOR') $|$ ($bin$ = $2$ \& $inv$ = $0$ \& $lib$ = `XOR') $|$ ($bin$ = $1$ \& $inv$ = 0 \& $lib$ = `XNOR') $|$ ($bin$ = $2$ \& $inv$ = $1$ \& $lib$ = `XNOR')} {
            $key\_value$ $\gets$ $1$\;
        }
        \KwRet $key\_value$\;
}  

  \SetKwProg{Fn}{function}{:}{\KwRet}
  \Fn{\FKR}{
        $\mathcal{C}_{cn}$ $\gets$ \texttt{ExtractLogicCone($\mathcal{C}_{lock}$,$KI$)}\;
        $\mathcal{C}_{cn\_syn}$ $\gets$ $\mathcal{S}$($\mathcal{C}_{cn}$)\;
        $PIP$ $\gets$ $\emptyset$\;
        \For{$key$ in $KI$ }{
            $idx$ $\gets$ $\mathit{get\_index(key,KI)}$\;
			$key\_value[idx]$ $\gets$ $\mathtt{x}$\;
			$gate$ $\gets$ $\mathit{gate\_conn}(key)$\;
			\{$bin$,$lib$,$inv$\} $\gets$ \texttt{GetAttribute($key$,$PIP$,$gate$)}\;
			$key\_value[idx]$ $\gets$ \texttt{KeyMapping($bin$,$lib$,$inv$)}\;    
        }
        \If{`$\mathtt{x}$' $\subseteq$ $key\_value$}{
            \texttt{K} $\gets$ $\mathit{SAT(\mathcal{C}_{lock},\mathbb{C}_{act},key\_value)}$\;
        } \Else{
            \texttt{K} $\gets$ $key\_value$
        }
        \KwRet \texttt{K}\;
  }
\caption{{Attack on Non-hard-coded Techniques}
\label{alg:KR_NHC}}
\end{algorithm}
\section{Experimental Investigation}
\label{sec:results}

In this section, we demonstrate the efficacy of our key-recovery attacks on the hardware implementation of
PSLL techniques across different parameters such as choice of (i)~technology library, (ii)~synthesis tool, (iii)~synthesis commands, (iv)~type of logic gates used for synthesis, and (v)~key-size.
In addition, we also elucidate some important findings.

\subsection{Experimental Setup}
\label{sec:setup}

\noindent\textbf{Locking techniques.} We implement all the PSLL techniques considered in this work using \textit{Perl} and \textit{Python} on three abstraction levels (\textit{BENCH}, RTL, and synthesized \textit{Verilog}).
We lock the circuits with a key-size of 128 for SARLock, SARLock-DTL, SFLL-HD$^0$, SFLL-flex, SFLL-rem, CAC, CAC-DTL, and ECE.
We choose a key-size of 256 for Anti-SAT, Anti-SAT-DTL, CASLock, SAS, and the variants of Gen-Anti-SAT.
The number of $\mathtt{PP}$ is 16 for SFLL-flex and we diversify the AND-tree by replacing 16 gates in the AND-tree with OR gates for DTL techniques.
We consider 4 SAS blocks for~\cite{liu2020strong}.

\noindent\textbf{Circuits.} We demonstrate the efficacy of our key-recovery attacks on eight combinational circuits from the ITC-99 suite.
We lock each circuit 100 times to capture variations in the selection of $PIPs$ and locking different $POPs$.

\noindent\textbf{Tool setup.} We perform synthesis of locked circuits using two commercial synthesis tools (\textit{Synopsys Design Compiler (DC)} and \textit{Cadence Genus}) and one academic synthesis tool (\textit{ABC}~\cite{brayton2010abc}).
We obtain $\mathtt{TPs}$ using an academic TPG tool, \textit{ATALANTA}~\cite{ATALANTA}.
We use two synthesis recipes,\footnote{Denotes a sequence of logic optimization commands.} \textit{synth\_A} and \textit{synth\_B}, when synthesizing circuits using \textit{Synopsys DC}.
While \textit{synth\_A} comprises \{\texttt{compile\_ultra}, \texttt{compile\_ultra -incremental}\}, \textit{synth\_B} comprises of three instantiations of \texttt{compile\_ultra} followed by three instantiations of \texttt{compile\_ultra -incremental}.
We use the command ``\texttt{synthesize -to\_mapped}'' with different efforts (medium and high) for \textit{Cadence Genus}.
Furthermore, we utilize six synthesis recipes (\textit{resyn}, \textit{resyn2}, \textit{resyn2a}, \textit{resyn3}, \textit{compress}, and \textit{compress2}) within the ABC tool~\cite{brayton2010abc} to verify the efficacy of our proposed attacks.

\noindent\textbf{Attack setup and evaluation metrics.} We implement our key-recovery attacks using \textit{TCL} and C++ scripts integrated with \textit{Synopsys DC} and \textit{ATALANTA}.
We use two metrics to assess the efficacy of our attacks, viz., (i)~accuracy, by computing the number of correct key-bits divided by the key-size and (ii)~precision, by checking the correctness of each key-bit.
We perform verification of the recovered key using a combinational equivalence checker within the \textit{ABC} tool.
We execute our attacks on a 128-core Intel Xeon processor running at 2.4 GHz having, 512 GB of RAM.

\subsection{Results of Key-Recovery Attacks}
\label{sec:KR_results}

\noindent\textbf{Applicability.} Our proposed attacks successfully recover the secret key (with accuracy and precision of 100\%) from
the hardware implementation of all 14 PSLL techniques considered in this work.
\textit{More importantly, our attacks highlight security vulnerabilities in nine previously unbroken locking techniques} (CAC~\cite{CAC}, CAC-DTL, SARLock-DTL, Anti-SAT-DTL~\cite{CAC}, SFLL-flex~\cite{yasin_CCS_2017}, ECE~\cite{shen2018comparative}, Gen-Anti-SAT (Comp. and Non-comp.)~\cite{zhou2021generalized}, and SAS~\cite{liu2020strong}).
We perform a detailed comparison with other key-recovery attacks in \S\ref{sec:comparison_prior_KR_attacks}.

\begin{table*}[ht]
\footnotesize
\setlength{\tabcolsep}{1.3mm}
\caption{Average execution time (in seconds) for \textbf{Key-recovery attack} across hundred random trials.
Locked circuits are synthesized using Synopsys DC with all logic gates from the Nangate 45nm library}
\label{tab:execution_time_KR_45nm}
\begin{tabular}{ccccccccclcccccc}
\hline
\multirow{3}{*}{\backslashbox{\textbf{Circuit}}{\textbf{Defense}}} &
\multicolumn{8}{c}{\textbf{Hard-coded}} &
& 
\multicolumn{6}{c}{\textbf{Non-hard-coded}}
\\ 
\cline{2-9} 
\cline{11-16} & \multirow{2}{*}{\textbf{SARLock}} & \multicolumn{3}{c}{\textbf{SFLL}} & \multirow{2}{*}{\textbf{CAC}} & \multirow{2}{*}{\textbf{ECE}} & \multicolumn{2}{c}{\textbf{DTL}} &
& 
\multirow{2}{*}{\textbf{Anti-SAT}} & \multirow{2}{*}{\textbf{CASLock}} & \multicolumn{2}{c}{\textbf{Gen-Anti-SAT}} & \textbf{DTL} & 
\multirow{2}{*}{\textbf{SAS}} 
\\ 
\cline{3-5} 
\cline{8-9} 
\cline{13-15}
&                                   
& 
\textbf{HD$^0$} & 
\textbf{flex} & 
\textbf{rem} &
&                               
& 
\textbf{SARLock} & 
\textbf{CAC} &
&  
&
& 
\textbf{Comp.} & 
\textbf{Non-comp.} & 
\textbf{Anti-SAT} &                         
\\ 
\hline 
\hline

\textbf{b14\_C} & 
35 & 
17 & 16 & 104 & 
22 & 
37 & 
34 & 111 & 
& 
34 & 
38 & 
44 & 
38 & 
32 & 
25 \\
\hline

\textbf{b15\_C} & 
42 & 
37 & 37 & 157 & 
43 & 
47 & 
43 & 128 & 
& 
37 & 
37 & 
48 & 
39 & 
34 & 
26 \\
\hline

\textbf{b20\_C} & 
83 & 
31 & 31 & 100 & 
52 & 
86 & 
83 & 138 & 
& 
67 & 
73 & 
110 & 
68 & 
73 & 
72 \\
\hline

\textbf{b21\_C} & 
80 & 
32 & 32 & 159 & 
58 & 
85 & 
81 & 
164 & 
& 
65 & 
73 & 
106 & 
70 & 
72 & 
69 \\
\hline

\textbf{b22\_C} & 
121 & 
39 & 37 & 212 & 
82 & 
123 & 
119 & 
161 & 
& 
93 & 
103 & 
168 & 
95 & 
87 & 
106 \\
\hline

\textbf{b17\_C} & 
158 & 
69 & 62 & 166 & 
122 & 
164 & 
158 & 
266 & 
& 
137 & 
131 & 
317 & 
141 & 
107 & 
133 \\
\hline

\textbf{b18\_C} & 
556 & 
361 & 365 & 405 & 
587 & 
570 & 
565 & 
743 &
& 
444 & 
540 & 
624 & 
620 & 
436 &
1,013 \\
\hline

\textbf{b19\_C} & 
637 & 
513 & 545 & 676 & 
557 & 
543 & 
668 & 
931 &
& 
681 & 
1,034 & 
868 &
966 &               
829 &
1,345 \\
\hline                 
\end{tabular}
\end{table*}

\begin{table*}[ht]
\footnotesize
\setlength{\tabcolsep}{1.3mm}
\caption{Average execution time (in seconds) for \textbf{Key-recovery attack} across hundred random trials.
Locked circuits are synthesized using Synopsys DC with all logic gates from GlobalFoundries 65nm library}
\label{tab:execution_time_KR_65nm}
\begin{tabular}{ccccccccclcccccc}
\hline
\multirow{3}{*}{\backslashbox{\textbf{Circuit}}{\textbf{Defense}}} &
\multicolumn{8}{c}{\textbf{Hard-coded}} &
& 
\multicolumn{6}{c}{\textbf{Non-hard-coded}} 
\\ 
\cline{2-9} 
\cline{11-16} & \multirow{2}{*}{\textbf{SARLock}} & \multicolumn{3}{c}{\textbf{SFLL}} & \multirow{2}{*}{\textbf{CAC}} & \multirow{2}{*}{\textbf{ECE}} & \multicolumn{2}{c}{\textbf{DTL}} &
& 
\multirow{2}{*}{\textbf{Anti-SAT}} & \multirow{2}{*}{\textbf{CASLock}} & \multicolumn{2}{c}{\textbf{Gen-Anti-SAT}} & \textbf{DTL} & 
\multirow{2}{*}{\textbf{SAS}} 
\\ 
\cline{3-5} 
\cline{8-9} 
\cline{13-15}
&                                   
& 
\textbf{HD$^0$} & 
\textbf{flex} & 
\textbf{rem} &
&                               
& 
\textbf{SARLock} & 
\textbf{CAC} &
&  
&
& 
\textbf{Comp.} & 
\textbf{Non-comp.} & 
\textbf{Anti-SAT} &
\\ 
\hline 
\hline

\textbf{b14\_C} & 
48 & 
20 & 20 & 98 &
28 &
50 &
46 & 114 &
& 
38 &
42 & 
47 & 
42 & 
35 &
75 \\
\hline

\textbf{b15\_C} & 
64 & 
43 & 53 & 134 & 
50 & 
65 & 
66 & 138 &
& 
37 & 
46 & 
48 & 
36 & 
34 &
98 \\
\hline

\textbf{b20\_C} & 
103 & 
31 & 41 & 113 & 
62 & 
104 & 
103 & 143 &
& 
73 & 
87 & 
96 & 
80 & 
72 &
115 \\
\hline

\textbf{b21\_C} & 
100 & 
33 & 40 & 178 & 
63 & 
106 & 
100 & 143 &
& 
75 & 
94 & 
87 & 
76 & 
67 &
120 \\
\hline

\textbf{b22\_C} & 
142 & 
40 & 49 & 243 & 
86 & 
140 & 
146 & 169 &
& 
96 & 
126 & 
137 & 
110 & 
85 &
153 \\
\hline

\textbf{b17\_C} & 
192 & 
71 & 83 & 276 & 
125 & 
202 & 
193 & 234 &
& 
110 & 
144 & 
169 & 
110 & 
109 &
224 \\
\hline

\textbf{b18\_C} &
643 &
395 &
401 & 527 & 659 &
614 &
670 &
789 &
&
629 &
611 &
599 &
602 &
615 &
971 \\
\hline

\textbf{b19\_C} &
667 &
579 &
511 & 632 & 711 &
579 &
723 &
887 &
&
777 &
1,161 &
907 &
911 &
1,120 &
1,459 \\
\hline
\end{tabular}
\end{table*}

\begin{table*}[htb]
\footnotesize
\setlength{\tabcolsep}{1.3mm}
\caption{Average execution time (in seconds) for \textbf{Key-recovery attack} across hundred random trials.
Locked circuits are synthesized using ABC with only 2-input AND gates and Inverters}
\label{tab:execution_time_KR_abc}
\begin{tabular}{ccccccccclcccccc}
\hline
\multirow{3}{*}{\backslashbox{\textbf{Circuit}}{\textbf{Defense}}} &
\multicolumn{8}{c}{\textbf{Hard-coded}} &
& 
\multicolumn{6}{c}{\textbf{Non-hard-coded}} 
\\ 
\cline{2-9} 
\cline{11-16} & \multirow{2}{*}{\textbf{SARLock}} & \multicolumn{3}{c}{\textbf{SFLL}} & \multirow{2}{*}{\textbf{CAC}} & \multirow{2}{*}{\textbf{ECE}} & \multicolumn{2}{c}{\textbf{DTL}} &
& 
\multirow{2}{*}{\textbf{Anti-SAT}} & \multirow{2}{*}{\textbf{CASLock}} & \multicolumn{2}{c}{\textbf{Gen-Anti-SAT}} & \textbf{DTL} & 
\multirow{2}{*}{\textbf{SAS}} 
\\ 
\cline{3-5} 
\cline{8-9} 
\cline{13-15}
&                                   
& 
\textbf{HD$^0$} & 
\textbf{flex} & 
\textbf{rem} &
&                               
& 
\textbf{SARLock} & 
\textbf{CAC} &
&  
&
& 
\textbf{Comp.} & 
\textbf{Non-comp.} & 
\textbf{Anti-SAT} &
\\ 
\hline 
\hline

\textbf{b14\_C} & 
36 & 
27 & 27 & 114 & 
27 & 
34 & 
35 & 133 & 
& 
42 & 
43 & 
53 & 
46 & 
39 &
72 \\
\hline

\textbf{b15\_C} & 
48 & 
54 & 56 & 147 & 
51 & 
45 & 
51 & 158 & 
& 
48 & 
51 & 
61 & 
51 & 
43 &
87 \\
\hline

\textbf{b20\_C} & 
70 & 
46 & 49 & 104 & 
45 & 
71 & 
72 & 158 & 
& 
81 & 
90 & 
98 & 
86 & 
77 &
103 \\
\hline

\textbf{b21\_C} & 
68 & 
48 & 50 & 133 & 
48 & 
72 & 
72 & 155 & 
& 
78 & 
87 & 
99 & 
84 & 
72 &
106 \\
\hline

\textbf{b22\_C} & 
92 & 
59 & 60 & 134 & 
59 & 
89 & 
90 & 
176 & 
& 
108 & 
116 & 
140 & 
114 & 
98 &
121 \\
\hline

\textbf{b17\_C} & 
124 & 
98 & 100 & 154 & 
102 & 
118 & 
129 & 
219 & 
& 
175 & 
205 & 
216 & 
179 & 
146 &
204 \\
\hline

\textbf{b18\_C} & 
160 & 
229 & 269 & 514 & 
230 & 
152 & 
165 & 
430 &                                
& 
255 & 
769 & 
312 & 
415 & 
306 &
929 \\
\hline

\textbf{b19\_C} & 
345 & 
551 & 449 & 917 & 
545 & 
347 & 
368 & 912 &
& 
444 & 
972 & 
784 & 
595 & 
958 &
1,101 \\
\hline
\end{tabular}
\end{table*}

\begin{table*}[ht]
\centering
\scriptsize
\setlength{\tabcolsep}{1.2mm}
\caption{Number of SAT-based attack Iterations (or oracle queries) required to recover undeciphered key-bits. Key-size is 128}
\label{tab:SAT_iterations_hard_coded}
\begin{tabular}{ccccccccccccccccccccccccc}
\hline
\multirow{3}{*}{\backslashbox{\textbf{Circuit}}{\textbf{Defense}}}
&
\multicolumn{4}{c}{\textbf{SARLock}} 
& 
\multicolumn{4}{c}{\textbf{SARLock-DTL}} 
& 
\multicolumn{4}{c}{\textbf{CAC}} 
&
\multicolumn{4}{c}{\textbf{CAC-DTL}} 
& 
\multicolumn{4}{c}{\textbf{ECE}} 
& 
\multicolumn{4}{c}{\textbf{SFLL-flex}} \\ \cline{2-25} 
& 
\multicolumn{2}{c}{\textbf{Synth\_A}} 
& 
\multicolumn{2}{c}{\textbf{Synth\_B}} 
& 
\multicolumn{2}{c}{\textbf{Synth\_A}} 
& 
\multicolumn{2}{c}{\textbf{Synth\_B}} 
& 
\multicolumn{2}{c}{\textbf{Synth\_A}} 
& 
\multicolumn{2}{c}{\textbf{Synth\_B}} 
& 
\multicolumn{2}{c}{\textbf{Synth\_A}} 
& 
\multicolumn{2}{c}{\textbf{Synth\_B}} 
& 
\multicolumn{2}{c}{\textbf{Synth\_A}} 
& 
\multicolumn{2}{c}{\textbf{Synth\_B}} 
& 
\multicolumn{2}{c}{\textbf{Synth\_A}} 
& 
\multicolumn{2}{c}{\textbf{Synth\_B}} \\ \cline{2-25} 
& 
\textbf{min} & \textbf{max} 
& 
\textbf{min} & \textbf{max}       
& 
\textbf{min} & \textbf{max}       
& 
\textbf{min} & \textbf{max}
& 
\textbf{min} & \textbf{max}       
& 
\textbf{min} & \textbf{max}       
& 
\textbf{min} & \textbf{max}       
& 
\textbf{min} & \textbf{max}       
& 
\textbf{min} & \textbf{max}       
& 
\textbf{min} & \textbf{max}       
& 
\textbf{min} & \textbf{max}       
& 
\textbf{min} & \textbf{max}  
\\ \hline

\textbf{b14\_C}   
& 
2 & 16
&
2 & 128
& 
2 & 16
&
2 & 32 
& 
2 & 32
&
2 & 8 
& 
2 & 8
&
2 & 64 
& 
2 & 16
&
2 & 32 
& 
2 & 64
&
4 & 128
\\ \hline

\textbf{b15\_C}   
& 
2 & 16
&
8 & 256
& 
2 & 16
&
2 & 128 
& 
2 & 2
&
2 & 8 
& 
2 & 8
&
2 & 128 
& 
2 & 16
&
8 & 64
& 
2 & 4
&
2 & 16
\\ \hline

\textbf{b20\_C}   
& 
2 & 16
&
2 & 64
& 
2 & 16
&
2 & 64 
& 
2 & 8
&
2 & 8  
& 
2 & 32
&
2 & 32 
& 
2 & 16
&
4 & 64 
& 
2 & 4
&
2 & 8
\\ \hline

\textbf{b21\_C}   
& 
2 & 16
&
2 & 64
& 
2 & 16
&
2 & 64 
& 
2 & 2
&
2 & 8 
& 
2 & 16
&
2 & 16 
& 
2 & 32
&
2 & 64
& 
2 & 2
&
4 & 32
\\ \hline

\textbf{b22\_C}   
& 
2 & 16
&
8 & 64
& 
2 & 16
&
2 & 64 
& 
2 & 2
&
2 & 8 
& 
2 & 32
&
2 & 64 
& 
2 & 32
&
4 & 64 
& 
2 & 16
&
2 & 64
\\ \hline

\textbf{b17\_C}   
& 
2 & 8
&
2 & 256
& 
2 & 16
&
2 & 128 
& 
2 & 2
&
2 & 16 
& 
2 & 16
&
2 & 64 
& 
2 & 16
&
2 & 32 
& 
2 & 4
&
2 & 16
\\ \hline

\textbf{b18\_C}   
& 
2 & 16
&
2 & 64
& 
2 & 16
&
2 & 64 
& 
2 & 2
&
2 & 8 
& 
2 & 16
&
2 & 64 
& 
2 & 16
&
4 & 128  
& 
2 & 4
&
2 & 4
\\ \hline
\textbf{b19\_C}   
& 
2 & 16
&
2 & 256
& 
2 & 16
&
2 & 256 
& 
2 & 2
&
2 & 8 
& 
2 & 16
&
2 & 32 
& 
2 & 16
&
4 & 512  
& 
2 & 4
&
2 & 8
\\ \hline
\end{tabular}
\end{table*}

\noindent\textbf{Execution time.} We document the execution time of our key-recovery attacks across 14 locking techniques for ITC-99 circuits in Table~\ref{tab:execution_time_KR_45nm}, Table~\ref{tab:execution_time_KR_65nm}, and Table~\ref{tab:execution_time_KR_abc}, respectively. 
We derive the execution time for each locking technique and locked circuit by averaging attack runtimes across 100 random trials.
We follow three setups, as explained next. 
First, we synthesize the locked circuits using \textit{Synopsys DC} with all Boolean logic gates available in a technology library (full-library) for \textit{Nangate 45nm} and \textit{GlobalFoundries 65nm}.
We document the average attack execution time for the aforementioned setup in Table~\ref{tab:execution_time_KR_45nm} and Table~\ref{tab:execution_time_KR_65nm}.
Next, we synthesize the locked circuits in a technology-agnostic manner by 
using an academic synthesis tool, \textit{ABC}~\cite{brayton2010abc}. 
We document the attack execution time for this setup in Table~\ref{tab:execution_time_KR_abc}.
Across all the considered PSLL techniques, benchmarks, and the aforementioned setups, our attack recovers the secret key (accuracy and precision of 100\%) in a maximum of 1,013 and 1,459 seconds for the two largest circuits (b18\_C with 117,941 gates and b19\_C with 237,962 gates) from the ITC-99 suite.

\noindent\textbf{Oracle-less versus Oracle-guided attacks.} Here, we provide further details regarding the efficacy of our key-recovery attacks by distinguishing whether the attack recovers the key in an \textit{oracle-less} or an \textit{oracle-guided} setting.
Recall that in an oracle-less setting, an attacker has access to only the locked circuit, while in an oracle-guided setting, an attacker has access to a working chip and the locked circuit (\S\ref{sec:threat_model}).
We depict the minimum (and maximum) number of SAT-based attack iterations required to recover the secret key for two different synthesis settings (\textit{synth\_A} and \textit{synth\_B}) for some PSLL techniques in Table~\ref{tab:SAT_iterations_hard_coded}.
Recall that we leverage an oracle (by launching the SAT-based attack~\cite{subramanyan15}) to recover the undeciphered key-bits (Alg.~\ref{alg:KR_HC} and~\ref{alg:KR_NHC}). 
On average, our attacks correctly recover a high percentage (94.5\%) of key-bits using only structural analysis of the locked circuit, \textit{i.e.,} in an oracle-less setting.
For instance, the maximum number of undeciphered key-bits across all locking techniques and circuits is nine for a key-size of 128 (Table~\ref{tab:SAT_iterations_hard_coded}).

Note that our key-recovery attacks do not target vulnerabilities in 
the underlying PSLL algorithm, rather we identify and leverage structural vulnerabilities in the hardware implementation of the considered PSLL techniques to recover the secret key.
Our attacks perform structural analysis of the locked circuit and recover the secret key by either (i)~leveraging $\mathtt{TPs}$ in hard-coded PSLL techniques or (ii)~utilizing attributes about $KIs$ in non-hard-coded PSLL techniques.
Since any Boolean function can be realized using different structural representations, it is imperative that we evaluate the efficacy of our attacks on locked circuits with varied structural representations.
Therefore, we perform a thorough analysis with regards to the choice of (i)~technology libraries (academic/commercial), (ii)~synthesis tools (academic/commercial), (iii)~synthesis commands, (iv)~logic gates used for synthesis, and (v)~protected pattern, for different circuits and PSLL techniques.
All the aforementioned parameters dictate the underlying structure (or graph-based representation) of locked circuits.
Although we do not showcase the attack execution time for all cases (due to limited space), we illustrate the distribution of oracle-less (oracle-guided attacks) for four PSLL techniques.

\begin{figure*}[htb]
\centering
\subfloat[]{\includegraphics[width=.24\textwidth]{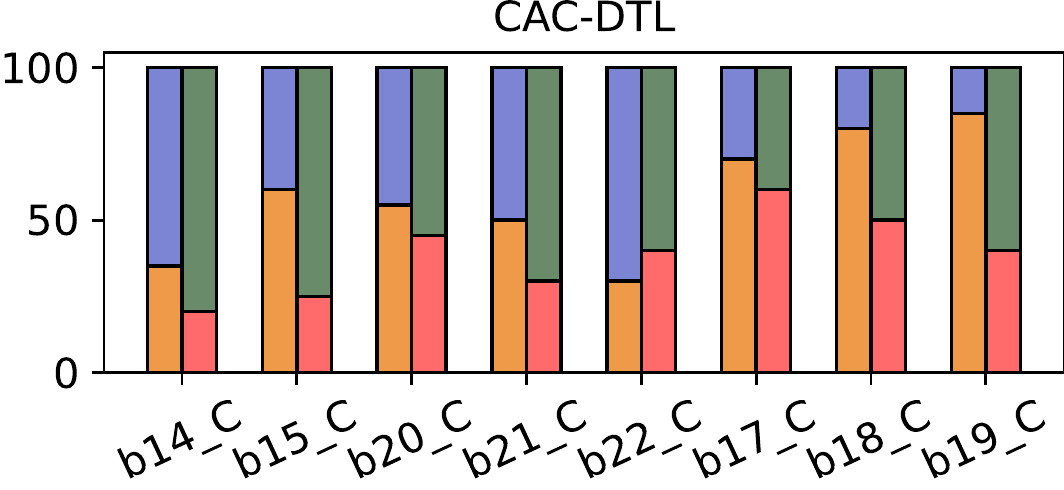}}
\hfill
\subfloat[]{\includegraphics[width=.24\textwidth]{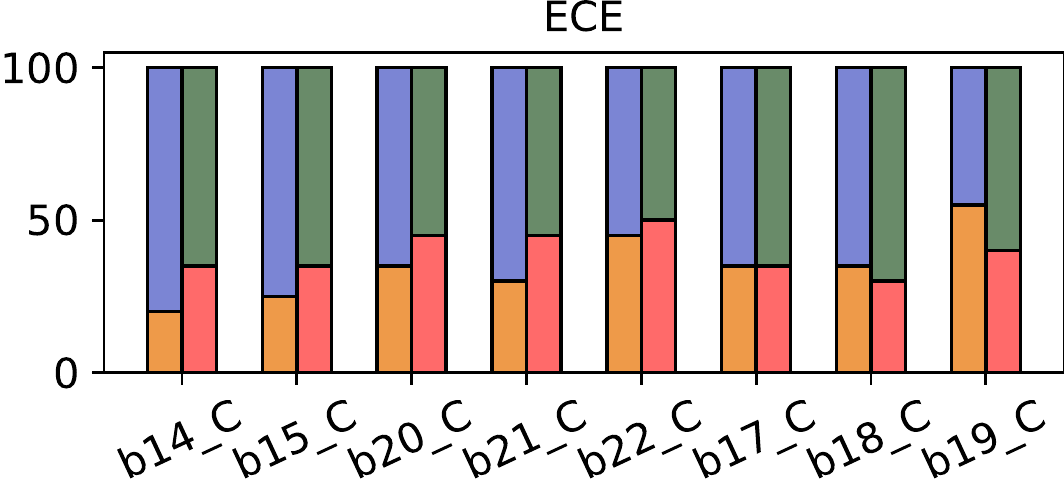}}
\hfill
\subfloat[]{\includegraphics[width=.24\textwidth]{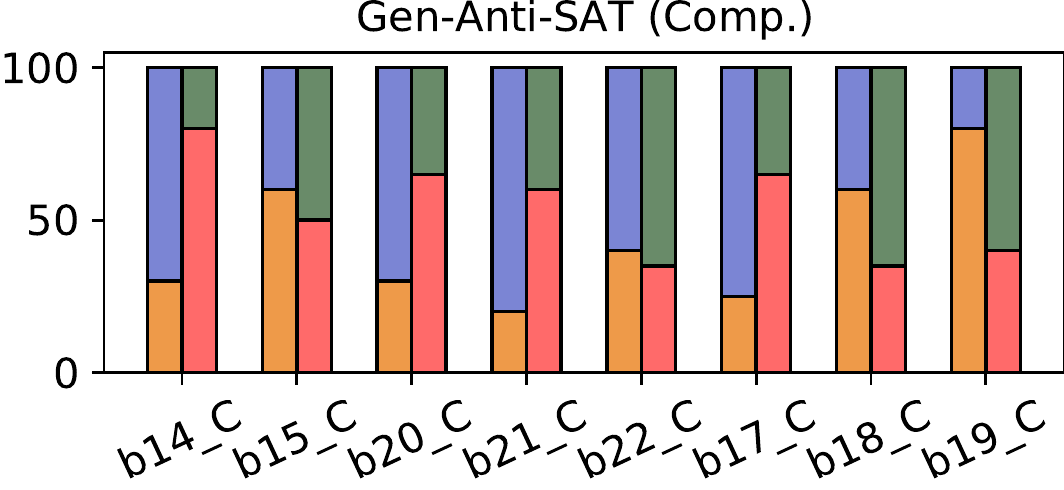}}
\hfill
\subfloat[]{\includegraphics[width=.24\textwidth]{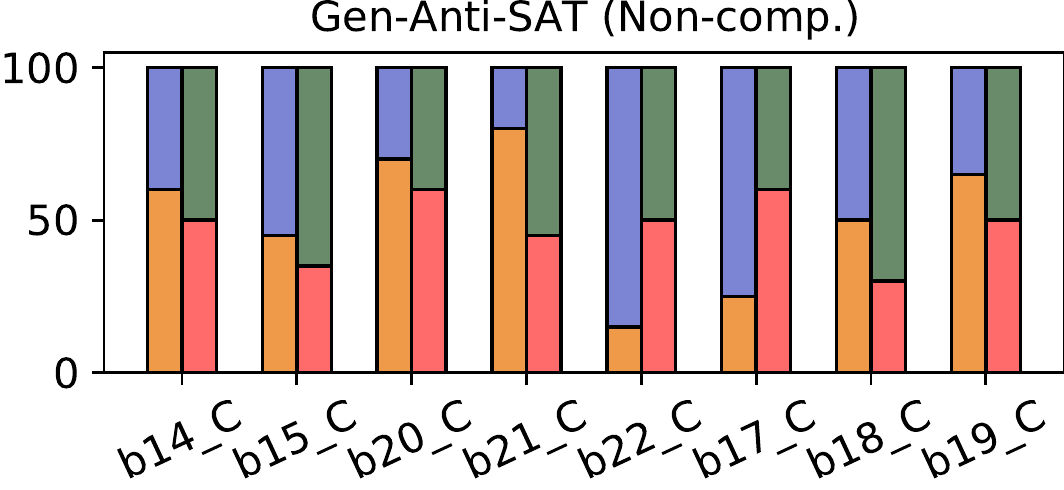}}
\caption{Distribution of oracle-less (oracle-guided) attacks for different technology libraries.
Orange (violet) denote oracle-less (oracle-guided) results for \textit{Nangate 45nm} technology library.
Red (green) denote oracle-less (oracle-guided) results for \textit{GlobalFoundries 65nm} technology library.}
\label{fig:myfig3}
\end{figure*}

\begin{figure*}[htb]
\centering
\subfloat[]{\includegraphics[width=.24\textwidth]{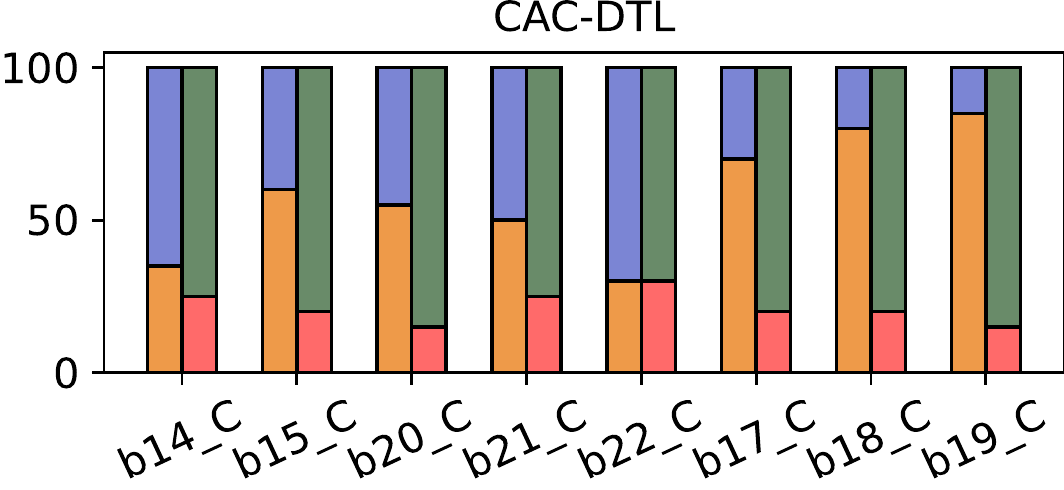}}
\hfill
\subfloat[]{\includegraphics[width=.24\textwidth]{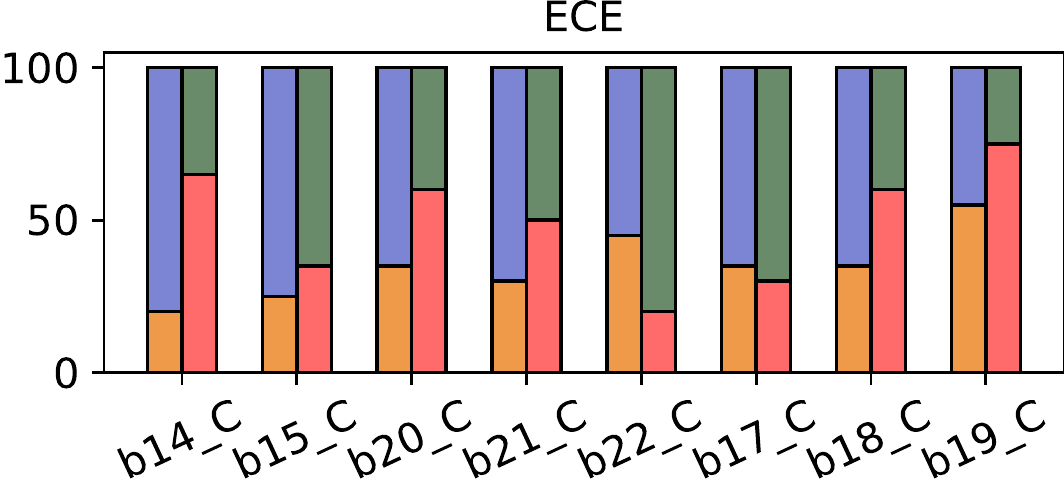}}
\hfill
\subfloat[]{\includegraphics[width=.24\textwidth]{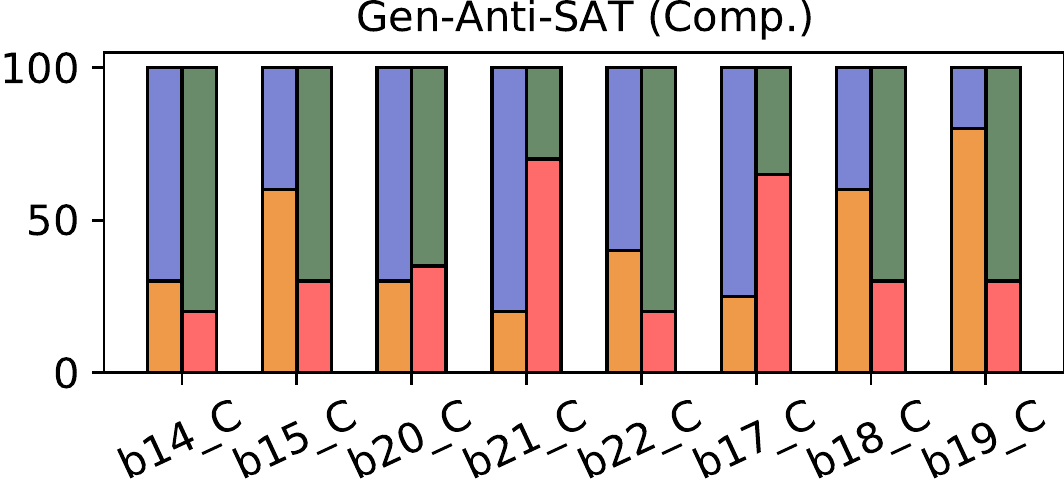}}
\hfill
\subfloat[]{\includegraphics[width=.24\textwidth]{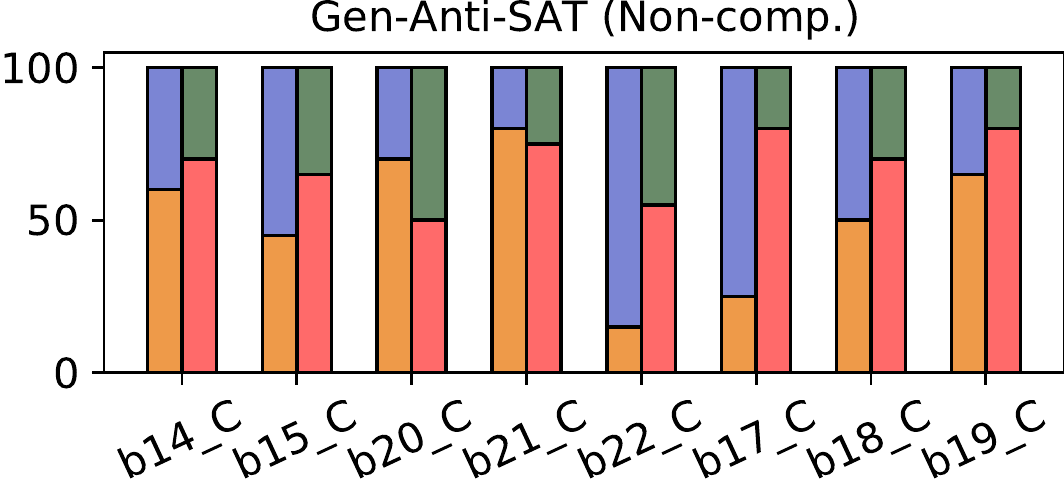}}
\caption{Distribution of oracle-less (oracle-guided) attacks for different synthesis tools. 
Orange (violet) denote oracle-less (oracle-guided) results for circuits synthesized using \textit{Synopsys DC}.
Red (green) denote oracle-less (oracle-guided) results for circuits synthesized using \textit{Cadence Genus}.}
\label{fig:myfig2}
\end{figure*}

\noindent\textbf{Effect of technology library.} To evaluate the effect of using different technology libraries (academic versus commercial), we synthesize the locked circuits using \textit{Synopsys DC} using only two-input gates with the same synthesis commands. 
The variable parameter is the technology library and the timing constraints for synthesis. 
We illustrate the distribution of oracle-less to oracle-guided attacks (stacked bar graphs) for four PSLL techniques in Fig.~\ref{fig:myfig3}. 
When locked circuits are synthesized using \textit{Nangate 45nm} library, our key-recovery attack recovers the secret key in an oracle-less setting in 50\% of cases for CAC-DTL, 35\% for ECE, 43.12\% for Gen-Anti-SAT (Comp.), and 51.25\% for Gen-Anti-SAT (Non-comp.).
These numbers are 38.75\% for CAC-DTL, 39.38\% for ECE, 53.75\% for Gen-Anti-SAT (Comp.), and 47.5\% for Gen-Anti-SAT (Non-comp.) when circuits are synthesized using \textit{GlobalFoundries 65nm} library.
The remaining circuits are broken in an oracle-guided setting.
\textbf{This analysis highlights the efficacy of our attacks across technology libraries.}

\begin{figure*}[t]
\centering
\subfloat[]{\includegraphics[width=.24\textwidth]{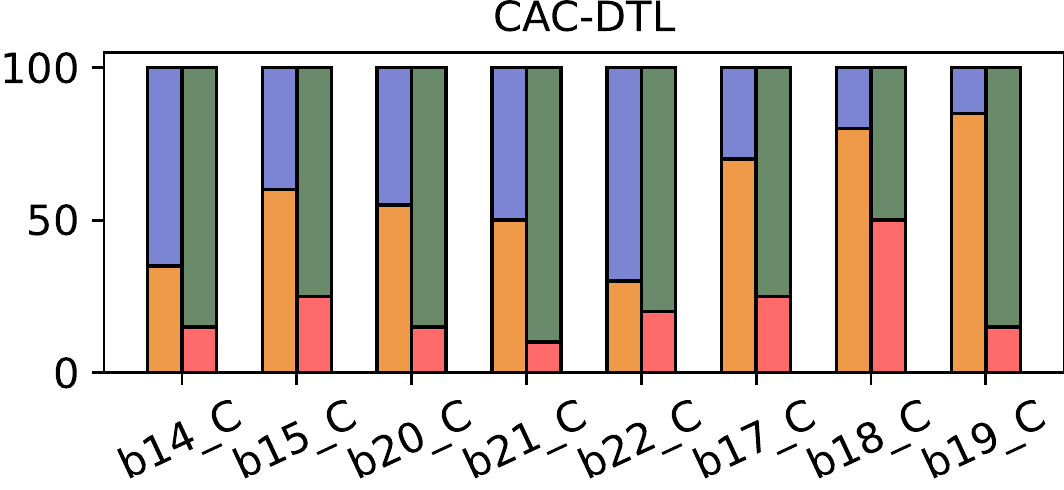}}
\hfill
\subfloat[]{\includegraphics[width=.24\textwidth]{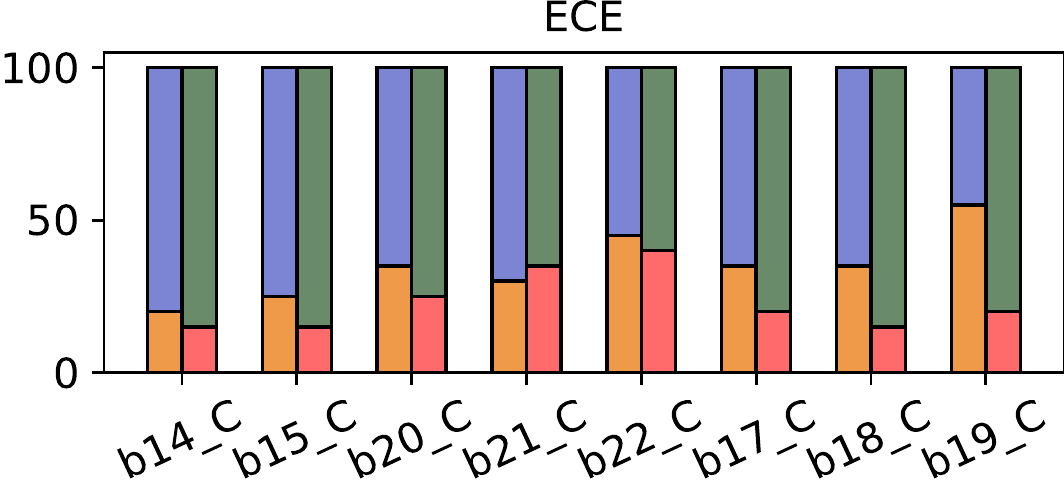}}
\hfill
\subfloat[]{\includegraphics[width=.24\textwidth]{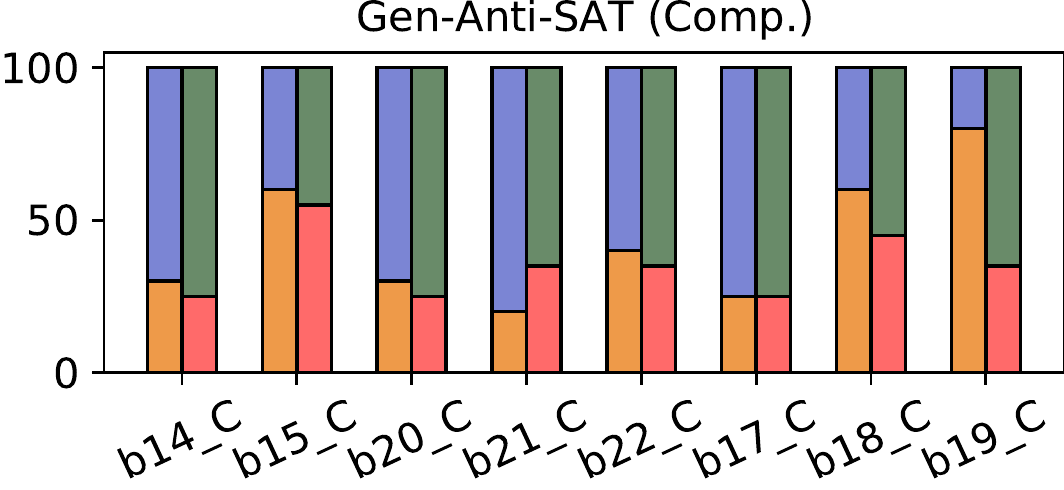}}
\hfill
\subfloat[]{\includegraphics[width=.24\textwidth]{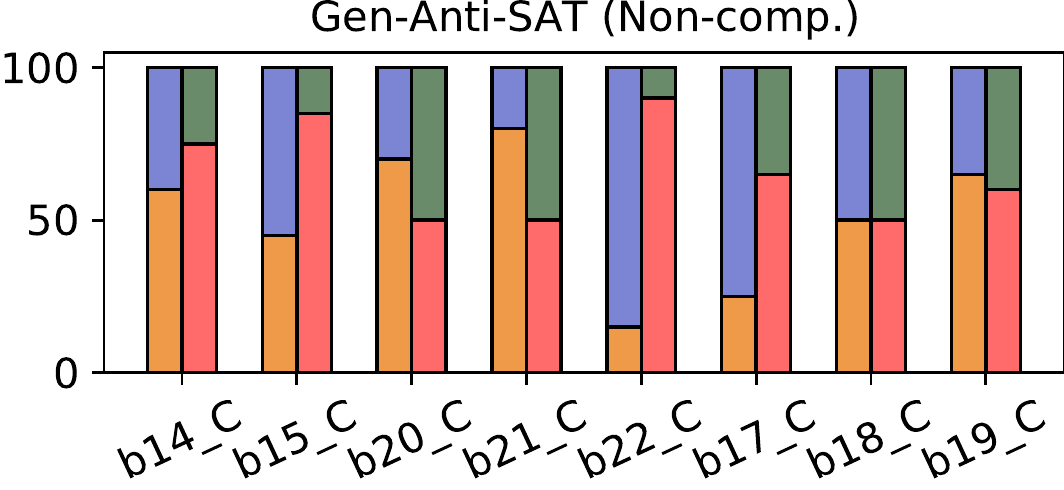}}
\caption{Distribution of oracle-less (oracle-guided) attacks for different type of gates used in synthesis. 
Orange (violet) denote oracle-less (oracle-guided) results for circuits synthesized using two-input gates.
Red (green) denote oracle-less (oracle-guided) results for circuits synthesized using full library.}
\label{fig:myfig1}
\end{figure*}

\begin{figure*}[htb]
\centering
\subfloat[]{\includegraphics[width=.24\textwidth]{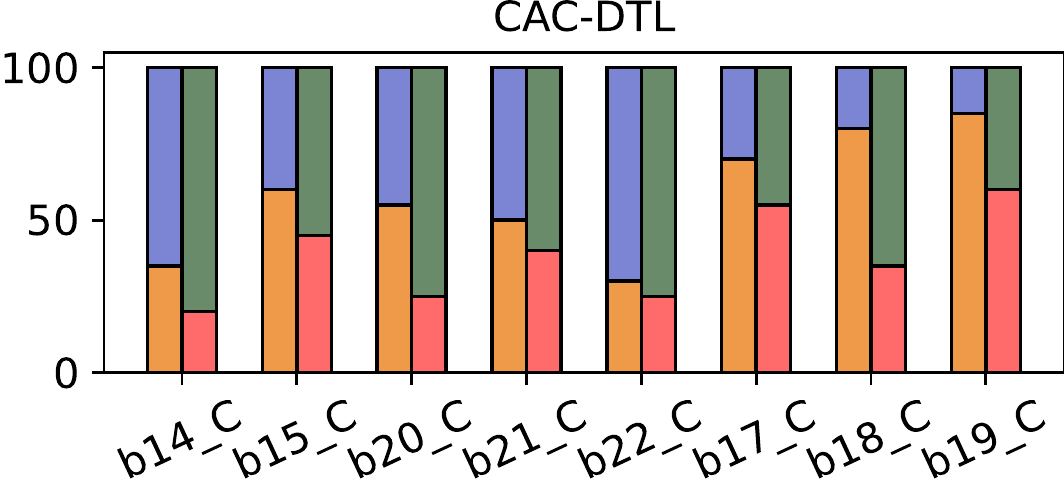}}
\hfill
\subfloat[]{\includegraphics[width=.24\textwidth]{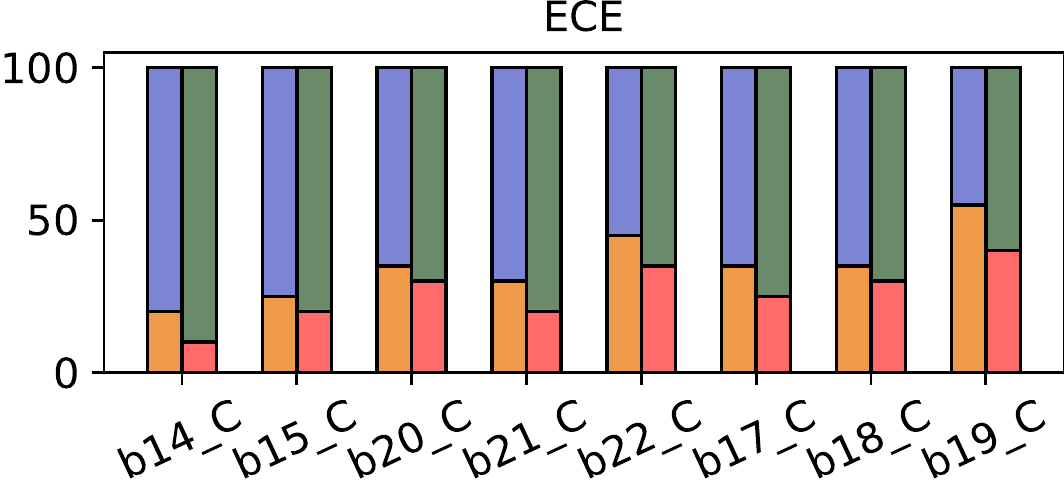}}
\hfill
\subfloat[]{\includegraphics[width=.24\textwidth]{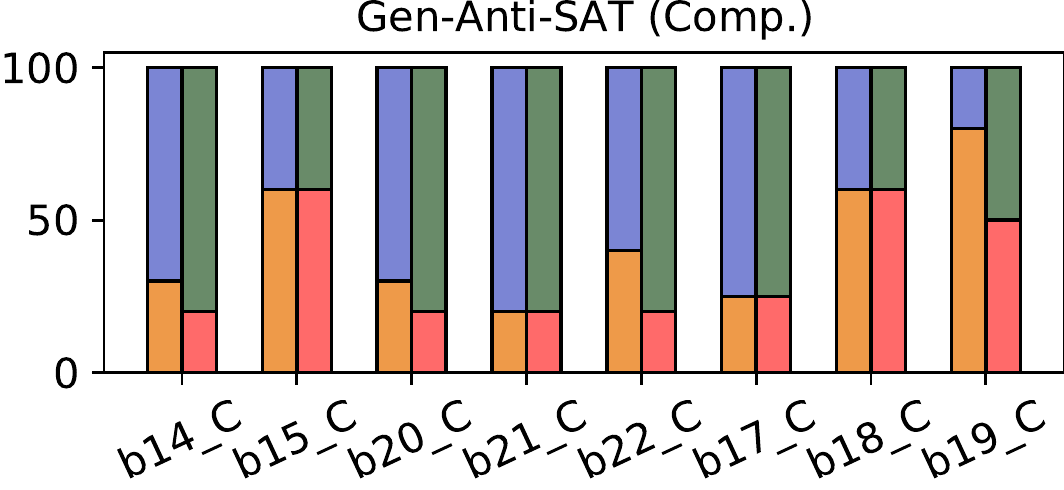}}
\hfill
\subfloat[]{\includegraphics[width=.24\textwidth]{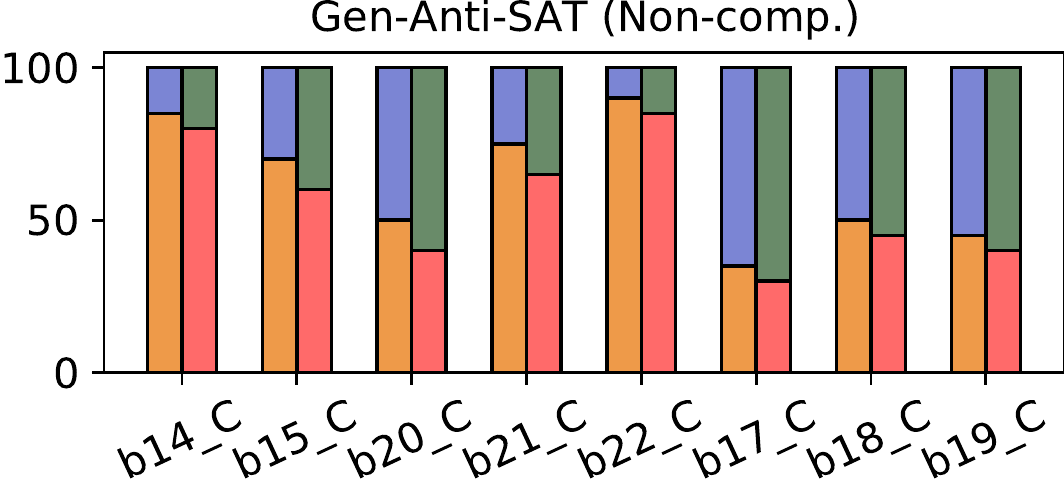}}
\caption{Distribution of oracle-less (oracle-guided) attacks for different synthesis settings. 
Orange (violet) denote oracle-less (oracle-guided) results for circuits synthesized using \textit{synth\_A}.
Red (green) denote oracle-less (oracle-guided) results for circuits synthesized using \textit{synth\_B}.}
\label{fig:OL_OG_PSLL_2inp_gates_synth_setting}
\end{figure*}

\begin{figure}[tb]
\centering
\subfloat[]{\includegraphics[width=.48\textwidth]{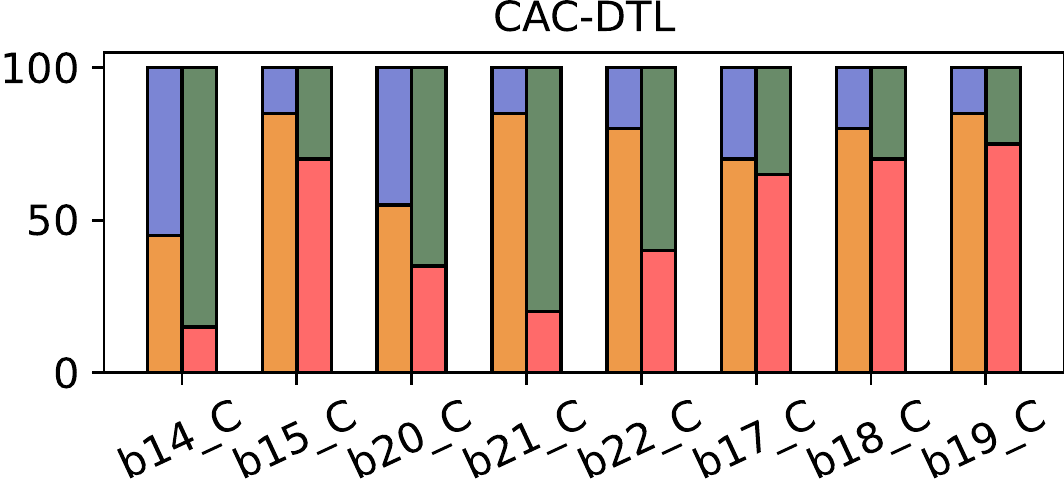}}\hfill
\subfloat[]{\includegraphics[width=.48\textwidth]{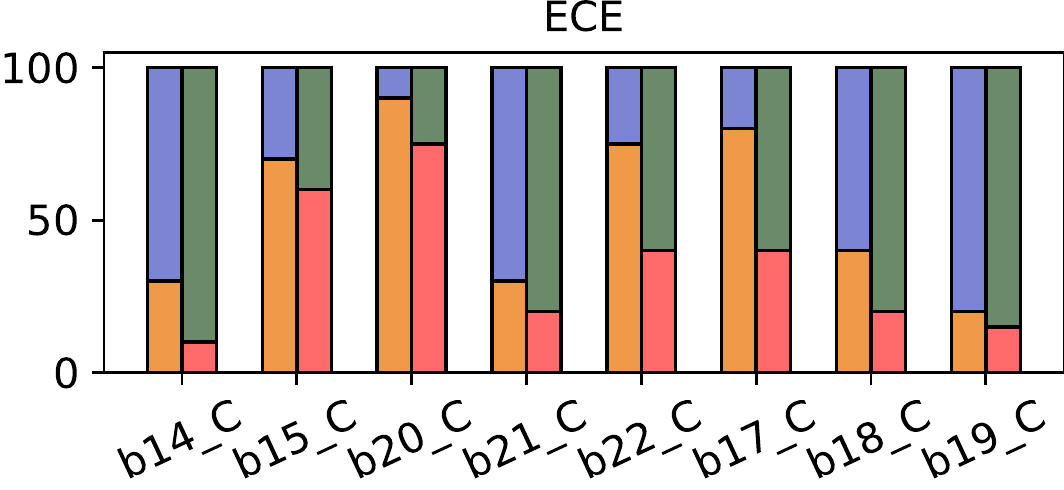}}\\
\caption{Distribution of oracle-less (oracle-guided) attacks for different $\mathtt{PP}$ for two hard-coded PSLL techniques.
Orange (violet) denote oracle-less (oracle-guided) results for circuits synthesized using \textit{synth\_A}.
Red (green) denote oracle-less (oracle-guided) results for circuits synthesized using \textit{synth\_B}.}
\label{fig:OL_OG_PSLL_PIPs_synth_setting}
\end{figure}

\begin{figure*}[ht]
\centering
\subfloat[]{\includegraphics[width=.24\textwidth]{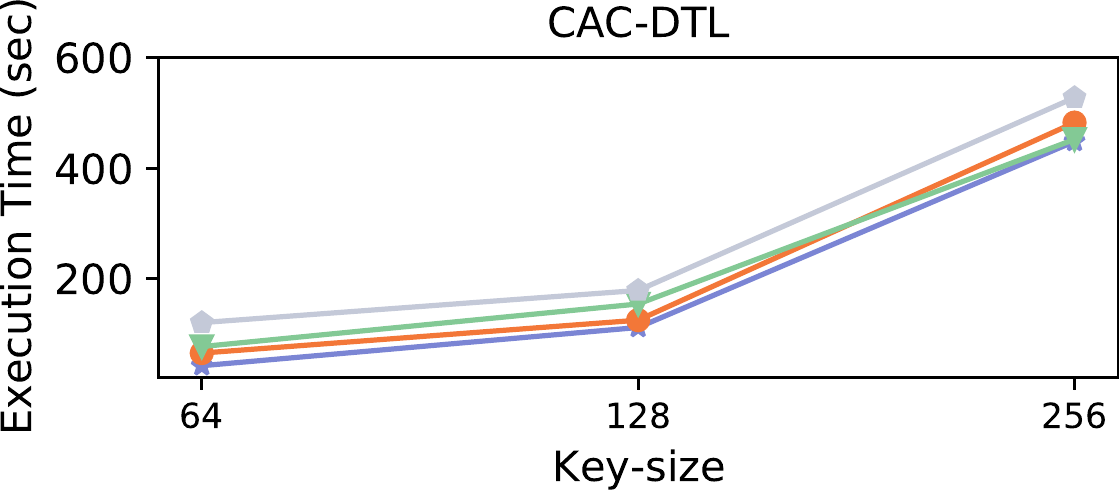}}
\hfill
\subfloat[]{\includegraphics[width=.24\textwidth]{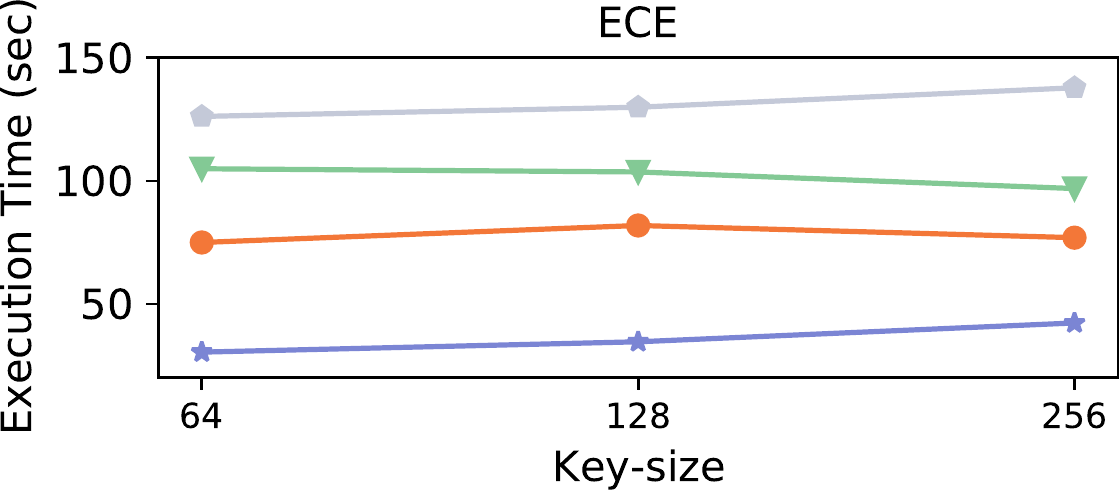}}
\hfill
\subfloat[]{\includegraphics[width=.24\textwidth]{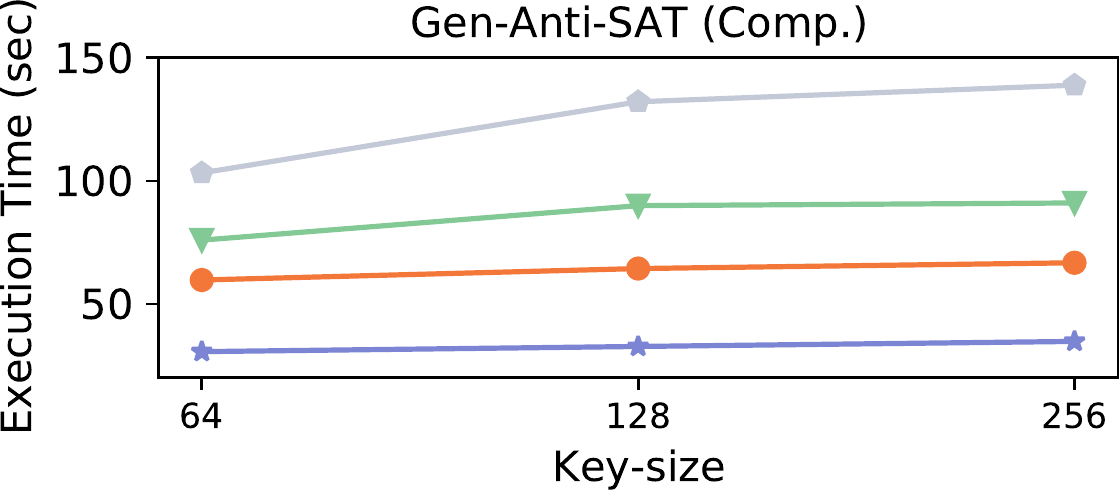}}
\hfill
\subfloat[]{\includegraphics[width=.24\textwidth]{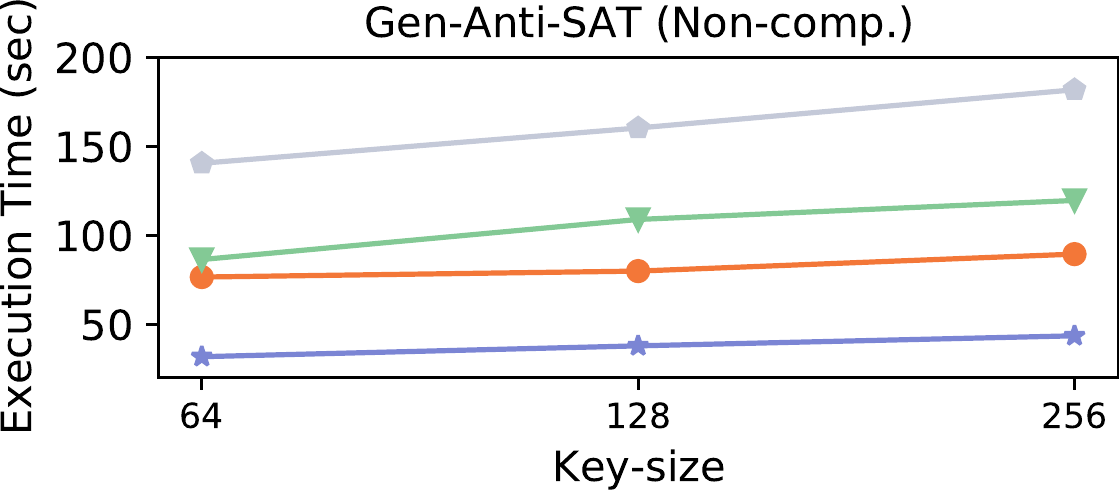}}
\caption{Effect of increasing key-size ($|K|$=64, $|K|$=128, and $|K|$=256) on the attack execution time for ITC-99 circuits (b14\_C, b15\_C, b17\_C, and b22\_C).
The gray, green, orange, and blue colors correspond to b17\_C, b22\_C, b15\_C, and b14\_C circuits (top to bottom), respectively.
}
\label{fig:key_size_execution_time}
\end{figure*}

\noindent\textbf{Effect of synthesis tool.} To evaluate the effect of using different synthesis tools, we fix the technology node (45nm) and the type of gates used for synthesis (2-input gates). 
We showcase the distribution of oracle-less to oracle-guided attacks in Fig.~\ref{fig:myfig2}. 
A majority of trials (78.75\% and 62.5\%) for CAC-DTL and Gen-Anti-SAT (Comp.) require oracle access when locked circuits are synthesized using \textit{Cadence Genus}. 
The percentage of trials requiring oracle access dropped to 43.13\% and 56.88\% when locked circuits are synthesized using \textit{Synopsys DC}.
This analysis demonstrates the efficacy of our attacks across different commercial synthesis tools.

\noindent\textbf{Effect of types of gates.} To evaluate the effect of utilizing different types of gates (2-input gates versus full-library) for synthesis, we fix the synthesis tool (\textit{Synopsys DC}), technology library (\textit{Nangate 45nm}), and synthesis commands (\texttt{compile\_ultra} +  \texttt{compile\_ultra -incremental}). 
We showcase the distribution of oracle-less and oracle-guided attacks in Fig.~\ref{fig:myfig1}. 
Most trials (78.13\%, 76.88\%, and 65\%) for CAC-DTL, ECE, and Gen-Anti-SAT (Comp.) require oracle access when synthesis is accomplished using the full library.
The percentage of trials requiring oracle access dropped to 41.88\%, 65\%, and 56.88\% when synthesis is performed using two-input gates.
The availability of diverse logic gates during full library synthesis enables the synthesis tool to optimize the structure of locked circuits. 
Such optimizations hinder the recovery of a \textit{few} key-bits when the attacker only has access to a locked circuit. 
However, an attacker can decipher the value of these unknown key-bits using an oracle.

\noindent\textbf{Effect of synthesis commands.} To evaluate the effect of synthesis commands, we fix the synthesis tool (\textit{Synopsys DC}), type of logic gates used for synthesis (2-input gates), and the technology library (\textit{Nangate 45nm}). 
We observe aggressive synthesis optimizations (\textit{synth\_B}) lead to more undeciphered key-bits (when an attacker attempts to recover the secret key in an oracle-less setting), thereby requiring access to an oracle to recover the secret key (Fig.~\ref{fig:OL_OG_PSLL_2inp_gates_synth_setting}).

\noindent\textbf{Effect of $\mathtt{PP}$.} To evaluate the impact of $\mathtt{PP}$, we fix the synthesis tool (\textit{Synopsys DC}), the type of logic gates used for synthesis (2-input gates), and the technology library (\textit{Nangate 45nm}).
The $PIPs$ and $POP$ are also kept constant across all locked circuits; the only variable parameter is the choice of $\mathtt{PP}$.
We illustrate the distribution of oracle-less to oracle-guided attacks for two hard-coded PSLL techniques in Fig.~\ref{fig:OL_OG_PSLL_PIPs_synth_setting}.
While we recover the secret key (accuracy of 100\%) in all locked circuits, we observe that the role of $\mathtt{PP}$ determines whether our attacks recover the secret key in an oracle-less setting (or not).
Furthermore, synthesis-induced logic optimizations, \textit{i.e.,} the use of different synthesis recipes also play a crucial role, as evidenced next.
While we recover the secret key in an oracle-less setting for 73.12\% (CAC-DTL) and 54.38\% (ECE) of the trials when locked circuits are synthesized using \textit{synth\_A}, these numbers drop to 48.75\% (CAC-DTL) and 35\% (ECE) when using \textit{synth\_B}.

\noindent\textbf{Effect of key-size.} We illustrate the impact of increasing key-size on the attack execution time in Fig.~\ref{fig:key_size_execution_time}. 
Increasing key-size does not affect the accuracy of our key-recovery attacks except for an increase in the attack execution time.

Note that, we do not identify anything wrong in the security proofs of the PSLL techniques considered in our work. 
The security proofs were directed primarily toward resilience against I/O-based attacks and assumed that an attacker would only query a working chip to recover the secret key. 
The scenario of an attacker aiming to recover the secret key \textit{only} through input/output pairs from a working chip corresponds to \textit{black-box cryptanalysis.}
Just as cryptographic algorithms provide security against an attacker with only black-box access to cryptographic devices, PSLL techniques provide provable-security against an attacker with only black-box access.
\textit{However, such a (black-box) model does not always correspond to the realities of hardware implementations.}
In a realistic and practical setting, an attacker has access to the locked circuit and the activated chip (\S\ref{sec:threat_model}) and attempt to recover the secret key through structural and functional analysis.

\llbox{
\textbf{Takeaway Message:} Our attacks successfully recover the secret key (with 100\% accuracy and 100\% precision) from the hardware implementation of all the considered PSLL techniques for all circuits across different parameters that include variations in (i)~technology libraries, (ii)~synthesis tools, (iii)~synthesis commands, (iv)~type of logic gates used, (v)~protected patterns, and (vi)~key-sizes.
}

\subsection{Important Observations from Key-Recovery Attacks}
\label{sec:findings_KR_attack}

\begin{enumerate}[leftmargin=*]

\item Our attack recovers the secret key (with 100\% accuracy and 100\% precision) in the following hard-coded PSLL techniques (SARLock, SARLock-DTL, ECE, CAC, CAC-DTL, SFLL-HD$^0$, SFLL-flex) in an \textit{oracle-less} setting when the defender synthesizes locked circuits using an academic synthesis tool \textit{ABC}~\cite{brayton2010abc}.
An exception is SFLL-rem, where the attack requires access to an oracle to recover some key-bits.
\textit{Recovering the secret key in an \textit{oracle-less} setting is powerful since it undermines the security guarantees of logic locking techniques during fabrication.}

\item Our KGM attack recovers the secret key (100\% accuracy and 100\% precision) for non-hard-coded PSLL techniques such as Anti-SAT, Anti-SAT-DTL, and Gen-Anti-SAT (Comp.) in an \textit{oracle-less} setting when circuits are synthesized using \textit{ABC}. 
In addition, the attack recovers the secret key for most cases for CASLock (92.5\%) and Gen-Anti-SAT (Non-comp.) (90\%) in an \textit{oracle-less} setting.
All locked circuits are broken using an oracle for SAS.

\item Our attacks recover the secret key for most PSLL techniques in an \textit{oracle-less setting} across six different synthesis recipes using \textit{ABC}. 
\textit{In a nutshell, synthesis operations carried out by \textit{ABC} do not aid in the dissolution of key-revealing logic gates(s) with the original circuit for both hard-coded and non-hard-coded PSLL techniques.}

\item Through our experiments, we observe that synthesis optimizations, usage of different technology libraries, etc., lead to a reduction of point-functions (Fig.~\ref{fig:KR_HC_eg}).
\textit{Our attack recovers the secret key successfully independent of procedures used by designers to generate the modified circuit (\S\ref{sec:PSLL}).}

\item To increase the output corruption of PSLL techniques and resist approximate attacks (\textit{Double-DIP}~\cite{shen2017double} and \textit{AppSAT}~\cite{shamsi2017appsat}), researchers proposed DTL~\cite{CAC} and ECE~\cite{shen2018comparative}. 
Our attacks recover the secret key for both techniques since the underlying construction either hard-codes the $\mathtt{PP}$ using (i)~a point-function (ECE), or (ii)~a point-function diversified with OR/NAND gates. 
The structural hints from key-revealing logic gate(s) (either point-function or reduced point-functions) are captured through $\mathtt{TPs}$, aiding an attacker to recover the secret key.

\item It has been established that hard-coded passwords in software artifacts are indicators of weakness in software security. 
For example, the CWE-798 mentions that ``\textit{if hard-coded passwords are used, it is almost certain that malicious users will gain access to the account in question.}'' 
Our key-recovery attacks attest to this from a hardware perspective, \textit{i.e.,} the hard-coding of secrets in PSLL techniques lead to structural vulnerabilities that attackers can exploit to recover the secret key.

\item Our experimental analysis reveals that the locking unit remains disjoint from the original circuit for all considered non-hard-coded PSLL techniques.

\item Finally, we observe that the choice of $\mathtt{PP}$ plays a role in the ability of a hard-coded point-function to merge with the original circuit.
During our experiments, we came across a few examples where the considered $\mathtt{PP}$ led to a significant dissolution of the point-function.
For instance, we observed that hard-coding a specific $\mathtt{PP}$ using a 64-input point-function led us to recover only 35 bits using our attack. 
This (empirical) finding highlights the role of $\mathtt{PP}$ toward structural security for hard-coded PSLL techniques.

\end{enumerate}

\subsection{Comparison with State-of-the-Art Key-Recovery Attacks}
\label{sec:comparison_prior_KR_attacks}

\noindent\textbf{Generality.} Recall that almost all the key-recovery attacks proposed by researchers cater towards a specific locking technique and do not generalize to other techniques (\S\ref{sec:research_challenges}).
On the other hand, our key-recovery attacks apply to a broad category of PSLL techniques (both hard-coded and non-hard-coded), including \textit{nine previously unbroken techniques}.

\noindent\textbf{Scalability.} Now we discuss the scalability of our attacks compared to prior key-recovery attacks.
We executed the FALL attack~\cite{FALL} on CAC-DTL, Gen-Anti-SAT, SFLL-rem, and SFLL-flex and observed that it was unsuccessful in recovering the secret key.
The key-bit mapping (KBM)-SAT attack~\cite{CASUnlock} does not work for all cascaded-chain configurations of AND/OR gates.\footnote{To evaluate the efficacy of KBM-SAT attack on cascaded-chain configurations of AND/OR gates, we chose a smaller key-size ($|K|$=16) to iterate over all possible configurations.
Then, we executed the attack over all configurations, \textit{i.e.,} $2^7$ configurations across seven AND/OR gates. 
We observed many configurations where KBM-SAT attack failed to recover the secret key. 
Notation-wise, let us represent the cascaded-chain configuration of AND-OR-AND as $\langle0,1,0\rangle$.
As per our experimental analysis for $|K|$=16, a cascaded-chain configuration $\langle1,0,0,0,1,1,1\rangle$, KBM-SAT attack requires at least $2^{|K|/4}-1$ iterations, where $|K|$ is 16. 
Let us denote the configuration $\langle1,0,0,0,1,1,1\rangle$ as $\langle1,0^m,1^m\rangle$, where $m$ is calculated as $16/4 - 1 = 3$.
Extrapolating this to a key-size of 128 (512) for the aforementioned configuration, the KBM-SAT attack will require \textit{at least} $2^{32}-1$ ($2^{128}-1$) iterations to recover the secret key.
We verified our hypothesis by executing the attack for seven days---KBM-SAT was unsuccessful in recovering the secret key.
However, our KGM attack successfully recovers the secret key in an \textit{oracle-less setting} for both key-sizes.}
We showcase the performance of the KBM-SAT and our attack for certain cascaded-chain configurations in Fig.~\ref{fig:kbm_sat_vs_kgm}.
On the contrary, our KGM attack recovers the secret key for all cascaded-chain configurations.
The SPI attack~\cite{SPI_USENIX} breaks SFLL-flex for only one $\mathtt{PP}$. 
However, SFLL-flex allows a designer to protect multiple $\mathtt{PP}$. 
Our attack recovers all $\mathtt{PP}$ for SFLL-flex in an \textit{oracle-less} setting.
Further, our attack is agnostic to the implementation style of the restore unit, \textit{i.e.,} whether constructed using look-up tables~\cite{yasin_CCS_2017} or logic gates.

\begin{figure}
\centering
\includegraphics[width=.75\columnwidth]{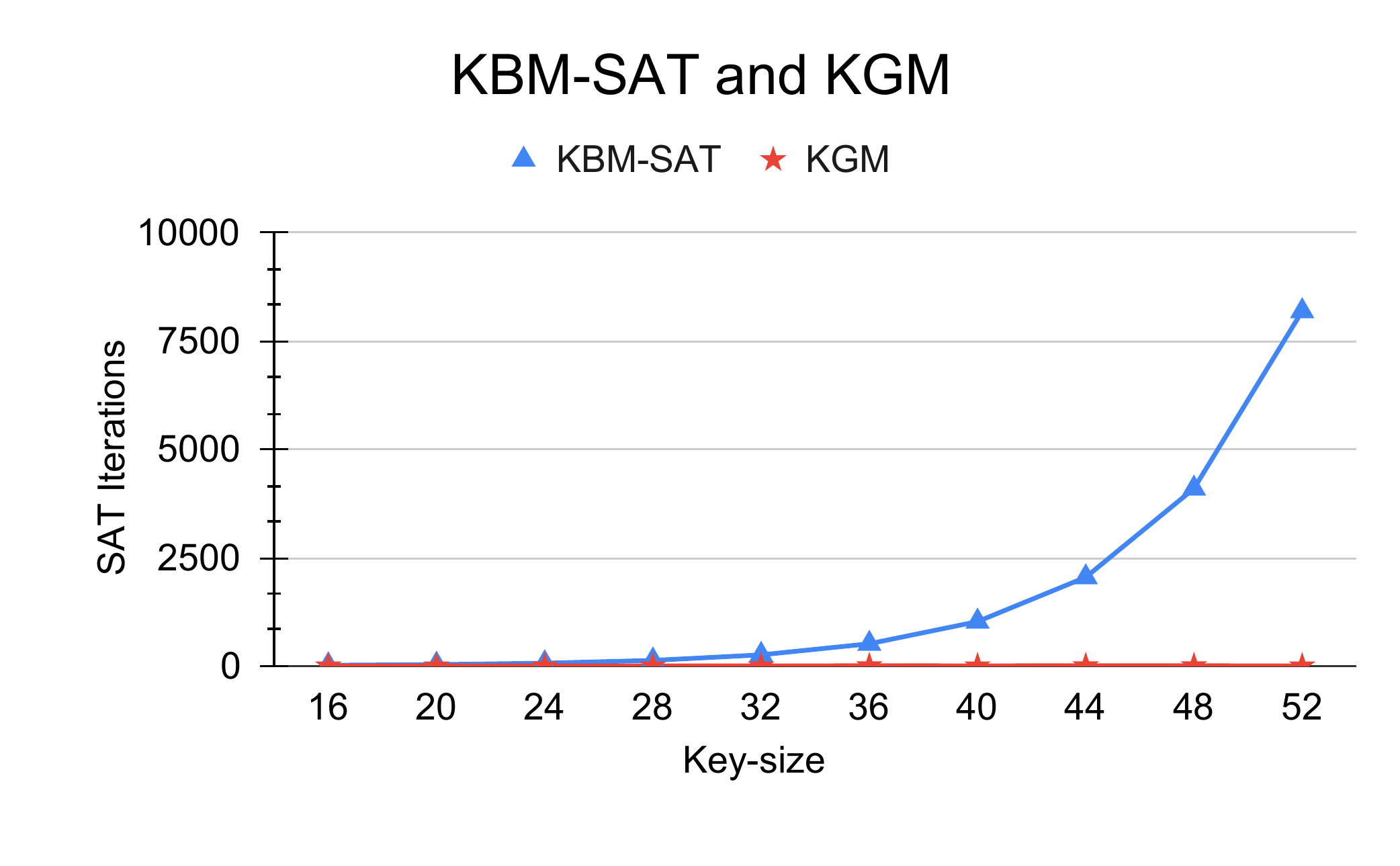}
\smallerspacecaption
\smallerspacecaption
\smallerspacecaption
\caption{Efficacy of key-bit mapping (KBM)-SAT attack~\cite{CASUnlock} compared with our KGM attack on b14\_C circuit locked using CASLock~\cite{caslock_ches_2019}.
}
\label{fig:kbm_sat_vs_kgm}
\end{figure}

\subsection{Application as a Diagnostic Tool for Designers}
\label{sec:diagnostic_tool}

Designers can utilize our key-recovery attacks to determine the lower bounds of security achieved by the hardware implementation of a PSLL technique.
Recall that our experimental analysis reinforces that implementing a PSLL technique (on hardware) with a key-size of $|K|$ does not necessarily imply $|K|$-bit security.
If an attacker recovers a portion of the secret key, $|K^*|$, then the actual security drops to $|K|-|K^*|$~\cite{guo2018introduction}.
Since the design IP is the ``secret sauce'' for industries and defense establishments, the ramifications are immense if the entity in question does not utilize a suitable diagnostic tool to ascertain the structural signatures emanated from the hardware implementation of a PSLL technique.
\section{Discussion}
\label{sec:discussion}

\subsection{Other IP Protection Solutions}
\label{sec:non-PSLL}

Design IPs can either be state-less (combinational) or state-full (sequential).
Researchers have proposed several logic locking techniques that lock the combinational portion of the design IP. 
Furthermore, since industrial design IPs are inherently sequential, sequential locking (sequential obfuscation) is another promising direction to prevent piracy of design IPs and unauthorized overproduction of ICs.
Sequential designs consist of finite state machines (FSMs), which are protected using FSM-based obfuscation. 
The original FSM is protected by (i)~inserting additional states known as obfuscated states~\cite{chakraborty2009harpoon} and black hole states~\cite{alkabani2007active}, (ii)~locking combinational logic cones of FSM states~\cite{janusHD_DATE},
and (iii)~re-encoding states to obscure the boundary between original state registers~\cite{TriLock_DATE}.
Hardware redaction is another approach where sensitive portions of the circuit are replaced with an embedded FPGA~\cite{mohan2021hardware}.
Another direction is eradicating key leakage in the scan mode~\cite{disorc,janusHD_DATE}.
Such techniques enable the use of high-corruption locking techniques.
Furthermore, provably secure block ciphers or pseudo-random functions can be adopted to thwart cryptanalysis and SAT-based attacks.

\subsection{Comparison with Password Selection and Cracking}
\label{sec:password_cracking_logic_locking}

\textit{A logic-locked circuit can be considered analogous to a password-protected computer system.}
In logic locking, the secret key is chosen either based on a seed or is a random string of $|K|$ bits. 
Unlike human-chosen passwords/PINs, the key search space is extensive (2$^{|K|}$) in logic locking~\cite{wang2017understanding}.
The key-recovery attacks discussed in our work aim to recover secret keys from the hardware implementation of PSLL techniques; this can be considered analogous to password cracking. 
Analogous to personally identifiable information in targeted online password guessing~\cite{wang2016targeted}, a partially recovered (correct) secret key from a locked circuit aids an attacker in reducing the security-level of a locked circuit.
However, unlike targeted password guessing~\cite{wang2016targeted}, there is no limit imposed on the number of queries an attacker makes to an oracle.

\subsection{Future Work}
\label{sec:future_work}

Our investigation highlights that vulnerabilities of the considered PSLL techniques stem from their hardware implementation. 
Although one can argue that 
synthesis tools prioritize
logic optimization over security-driven goals, our observations about specific $\mathtt{PP}$ that lead to the dissolution of key-revealing logic gate(s)
are important.
This means an interplay exists between the structure (and/or functionality) of the underlying circuit and the $\mathtt{PP}$ we wish to protect. 
Searching for this elusive set of $\mathtt{PP}$ is computationally challenging due to the exponential complexity of the number of input patterns. 
Recall that our analysis of non-hard-coded PSLL techniques
reveals that the locking unit remains 
disjoint from the original circuit.
Addressing the (i)~dissolution of the locking unit, (ii)~integration of security algorithms with CAD tools, and (iii)~extension of attacks presented in this paper toward sequential obfuscation techniques is reserved for future work.
\section{Conclusion}
\label{sec:conclusion}

In this work, we conceptualize and develop generalized attacks that recover the secret key from the hardware implementation of PSLL techniques. 
We extract various structural and functional properties contingent on the underlying hardware implementation and use (i)~principles of test pattern generation and (ii)~Boolean transformations to develop two attack algorithms that recover the secret key from locked circuits.
We evaluate the efficacy of our attacks across different parameters such as the choice of technology libraries, synthesis tools, synthesis settings, key-size, etc., and observe 100\% accuracy for 14 PSLL techniques (including nine previously unbroken techniques).
Besides demonstrating the security-obliviousness of current academic and commercial computer-aided design tools, our attacks provide several important insights, viz., (i)~the structural security of hard-coded PSLL techniques are contingent on the choice of the secret key, and (ii)~the locking unit remains disjoint from the original circuit for non-hard-coded PSLL techniques.
Additionally, developers of PSLL techniques can utilize our attacks 
to ascertain the lower-bound security 
achieved by hardware implementations.
Finally, we release our locked circuits and attack binaries with the hope that (i)~it would foster the development of secure hardware implementation of PSLL techniques and (ii)~future attackers can benchmark the performance of their developed attacks on common datasets, thereby enabling reproducibility.

\def\bibfont{\footnotesize}
\bibliography{main}
\bibliographystyle{IEEEtran}

\end{document}